%% file: wrapper.tex
\documentclass[11pt]{article}  
\usepackage{natbib}
\usepackage{amsmath}
\usepackage{amssymb}
\usepackage{amsthm}
\usepackage{graphicx}
\usepackage{tabularx}
\usepackage[table]{xcolor}
\usepackage{algpseudocode}
\usepackage{algorithm}
\usepackage{listings}
\usepackage{setspace}
\usepackage{hyperref}
\usepackage[margin=1in]{geometry}

\theoremstyle{definition}
\newtheorem{thm}{Theorem}[section]
\newtheorem{lemma}{Lemma}

\newcommand{\kdet}{{\tt kdetrees}}
\newcommand{\pmc}{{\tt Phylo-MCOA}}

\newcommand{\R}{\mathbb R}

\newcommand{\kde}{{\scshape kdetrees}}
\newcommand{\Atwo}{A_2(\theta_1,\theta_2)}

\title{Normalizing kernels in the Billera-Holmes-Vogtmann treespace}
\author{Grady Weyenberg \and Daniel K Howe  \and Ruriko Yoshida}
\date{}

\begin{document}

\maketitle
\newcommand{\sect}{manuscript}
\begin{abstract}

\noindent
{\bf Motivation:} 
As costs of genome sequencing have dropped precipitously, development
of efficient bioinformatic methods to analyze genome structure and
evolution have become ever more urgent. For example, most published
phylogenomic studies involve either massive concatenation of
sequences, or informal comparisons of phylogenies inferred on a small
subset of orthologous genes, neither of which provides a comprehensive
overview of evolution or systematic identification of genes with
unusual and interesting evolution (e.g. horizontal gene transfers,
gene duplication and subsequent 
neofunctionalization).  We are interested in identifying such ``outlying'' gene
trees from the set of gene trees and estimating the distribution of
the tree over the ``tree space''.

\noindent
{\bf Results:}
This paper describes an improvement to the
\kde{} algorithm, an adaptation of classical kernel density estimation to the metric
space of phylogenetic trees (Billera-Holmes-Vogtman treespace),
whereby the kernel normalizing constants, are estimated 
through the use of the novel holonomic gradient methods.  As the
original \kdet{} paper, we have applied 
\kdet{} to a set of 
Apicomplexa genes and it identified several unreliable sequence alignments
which had escaped previous detection, as well as a gene independently
reported as a possible case of horizontal gene transfer.

\noindent
{\bf Availability:}
The updated version of the \kde{} software package is available both
from CRAN (the official R package system), as well as from the
official development repository on Github. (\url{github.com/grady/kdetrees}).

\noindent
{\bf Contact:}
\href{ruriko.yoshida@uky.edu}{ruriko.yoshida@uky.edu} 
\end{abstract}

\input{new}

\bibliographystyle{named}
\bibliography{dissertation}

\appendix

\begin{figure}
  \centering
\includegraphics[width=\textwidth]{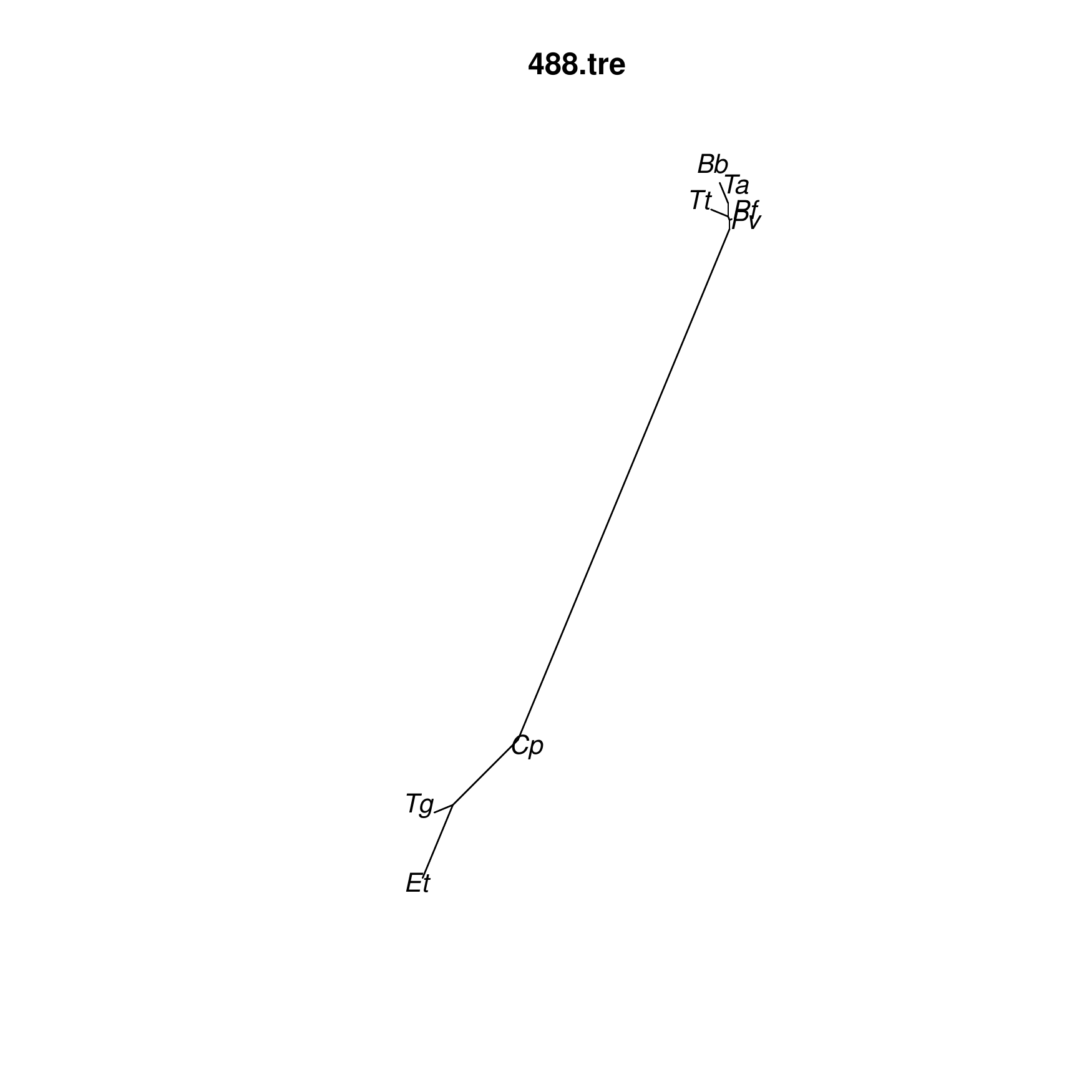}
\caption{A newly identified outlier from the Apicomplexa dataset.}

\end{figure}
\begin{figure}
  \centering
\includegraphics[width=\textwidth]{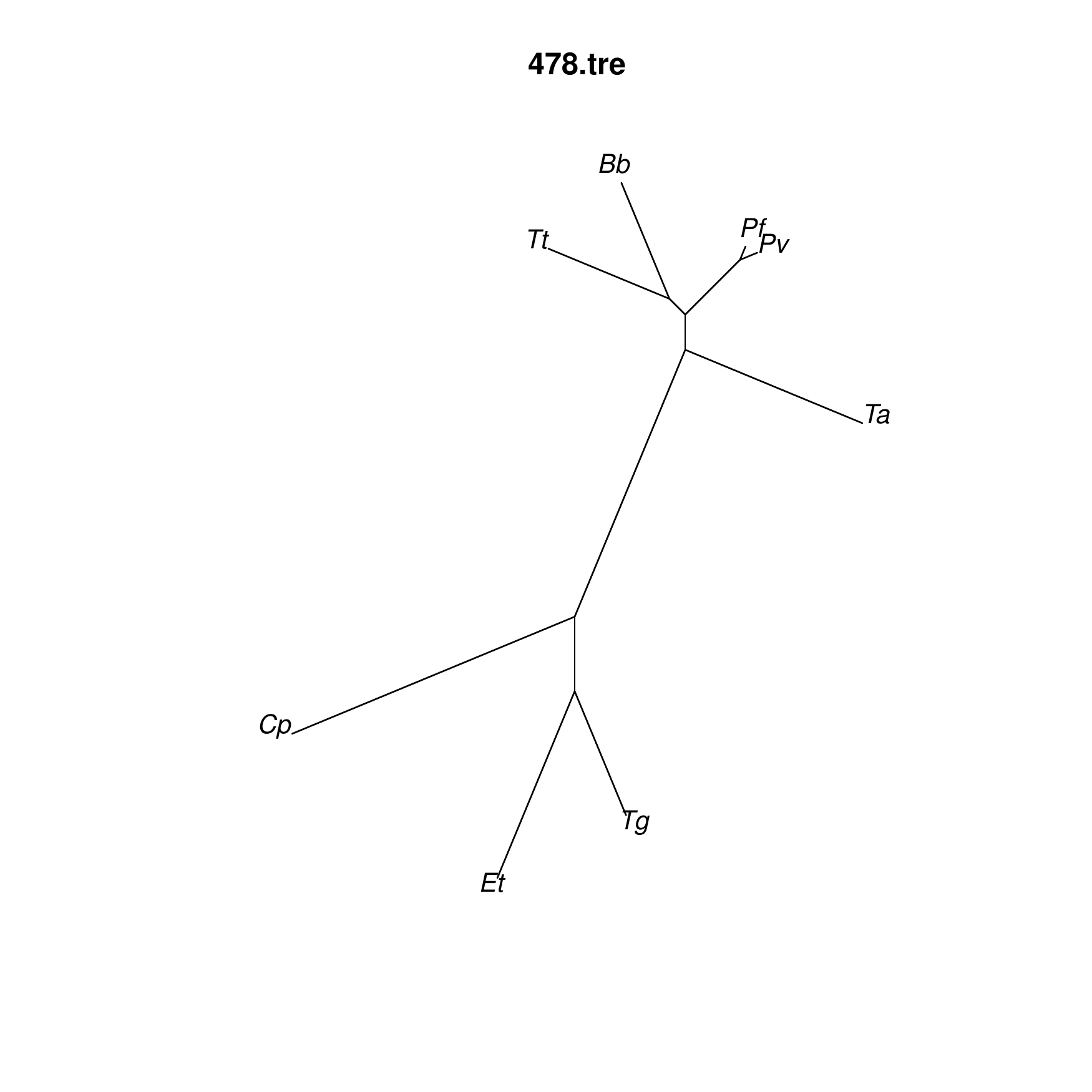}
\caption{A newly identified outlier from the Apicomplexa dataset.}
\label{fig:apiout}
\end{figure}
\begin{figure}
  \centering
\includegraphics[width=\textwidth]{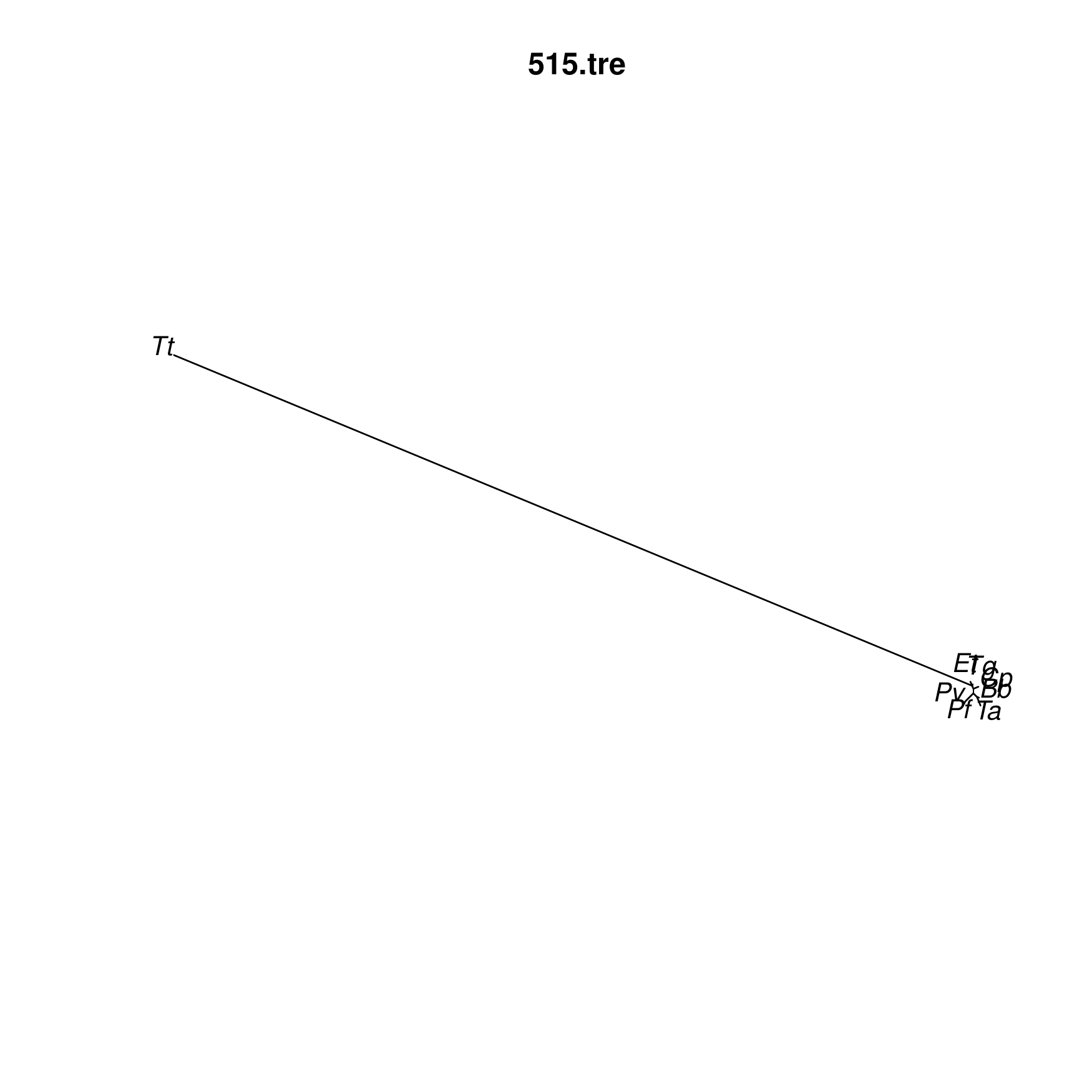}
\caption{A newly identified outlier from the Apicomplexa dataset.}

\end{figure}
\begin{figure}
  \centering
\includegraphics[width=\textwidth]{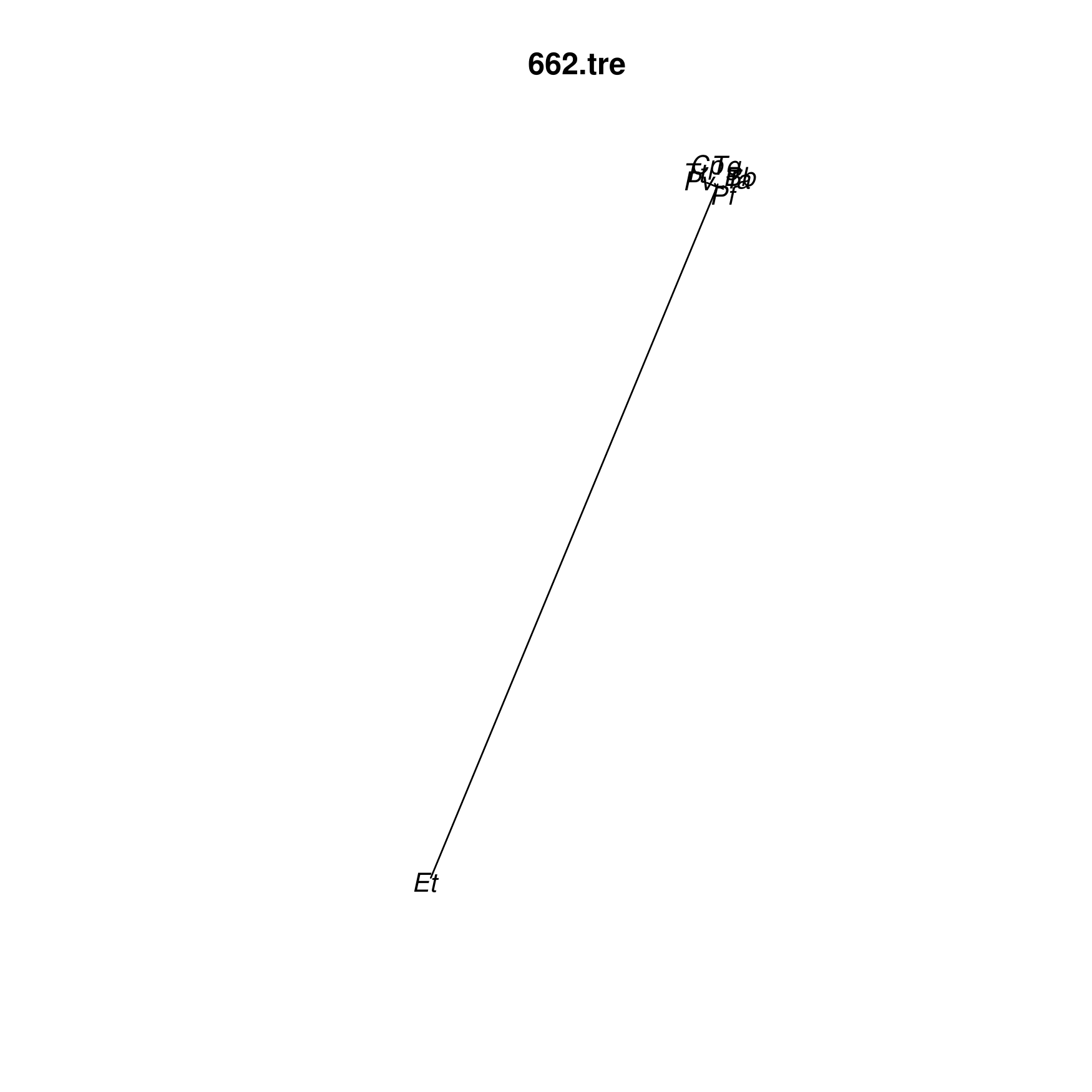}
\caption{A newly identified outlier from the Apicomplexa dataset.}

\end{figure}
\begin{figure}
  \centering
\includegraphics[width=\textwidth]{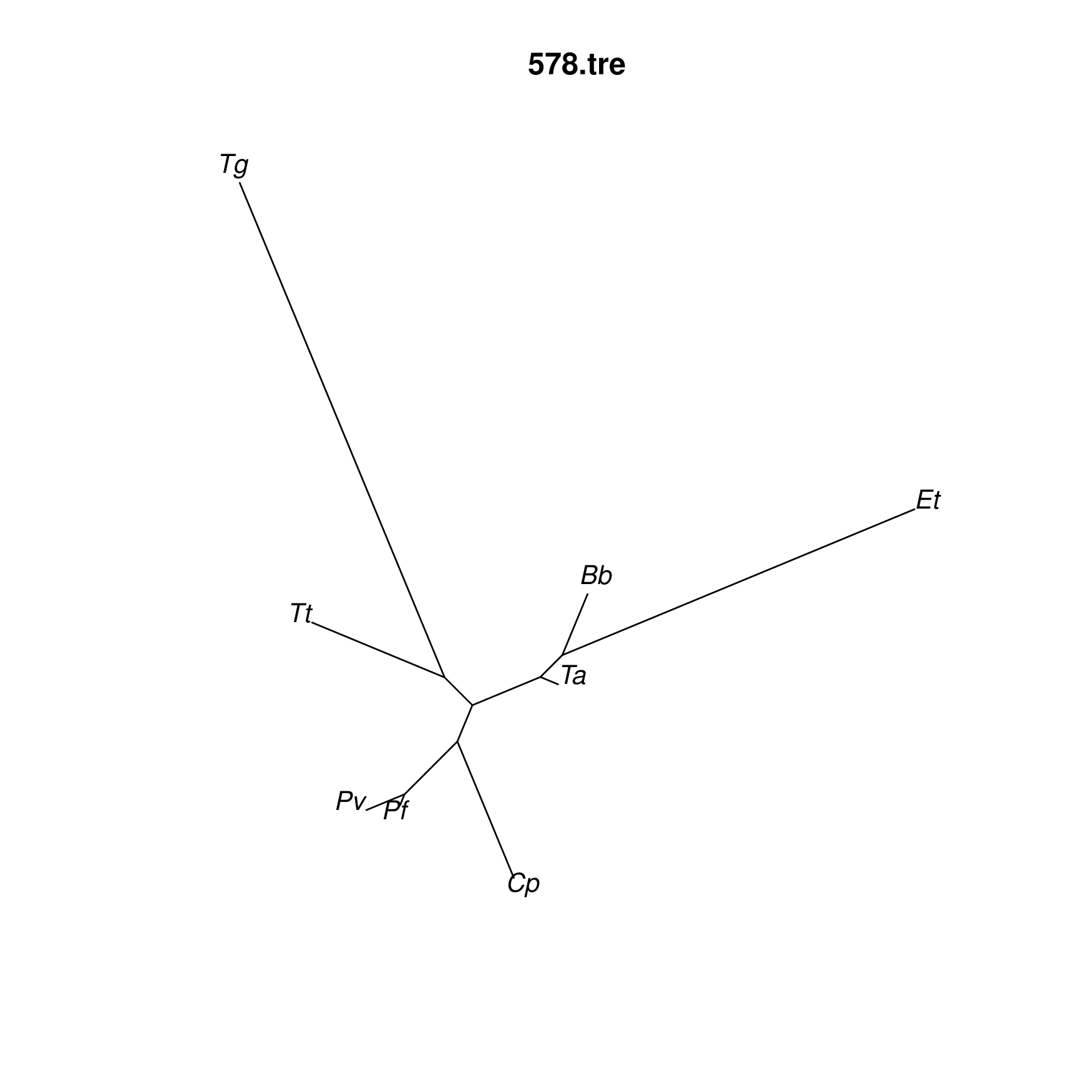}
\caption{A newly identified outlier from the Apicomplexa dataset.}

\end{figure}
\begin{figure}
  \centering
\includegraphics[width=\textwidth]{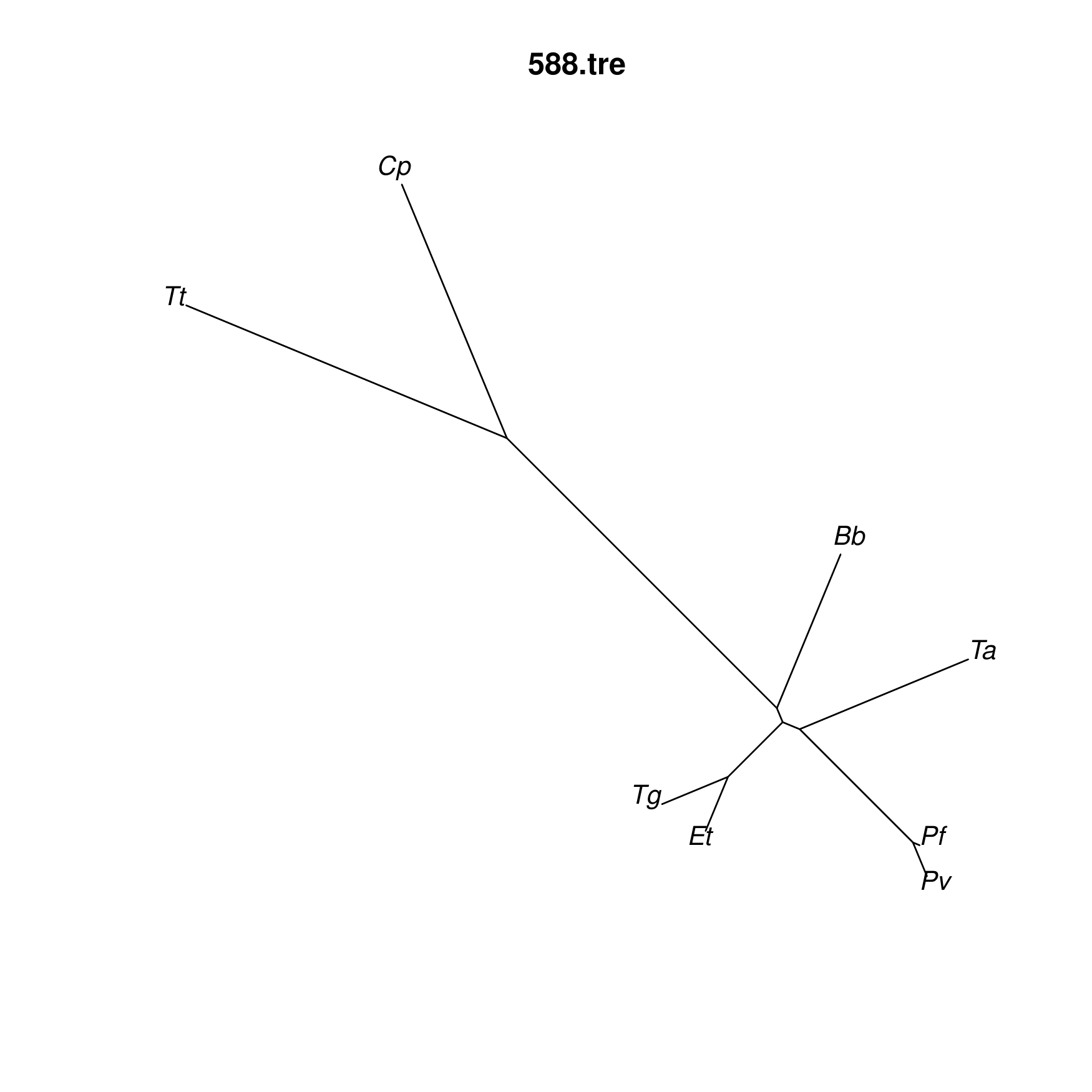}
\caption{A newly identified outlier from the Apicomplexa dataset.}

\end{figure}
\begin{figure}
  \centering
\includegraphics[width=\textwidth]{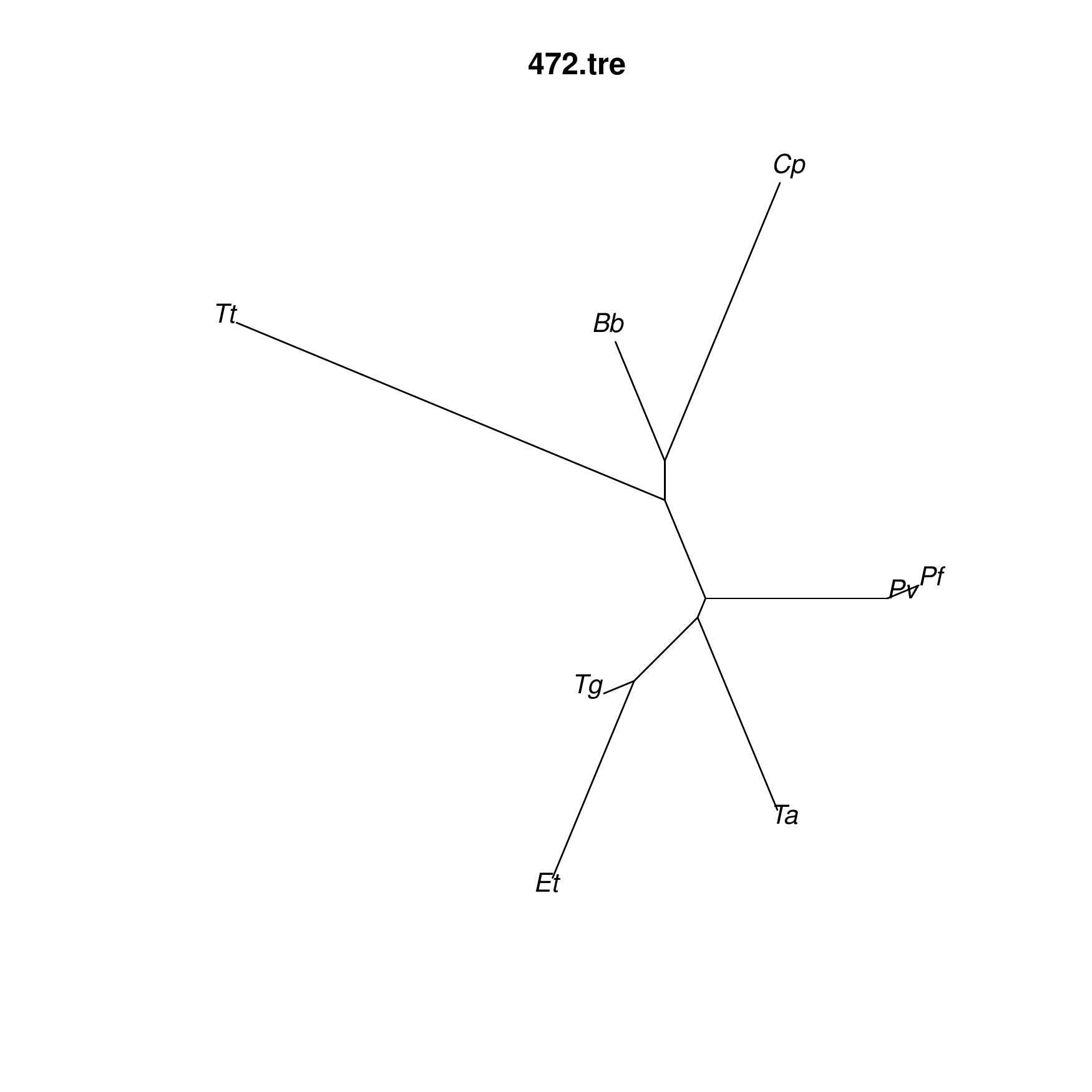}
\caption{A newly identified outlier from the Apicomplexa dataset.}

\end{figure}
\begin{figure}
  \centering
\includegraphics[width=\textwidth]{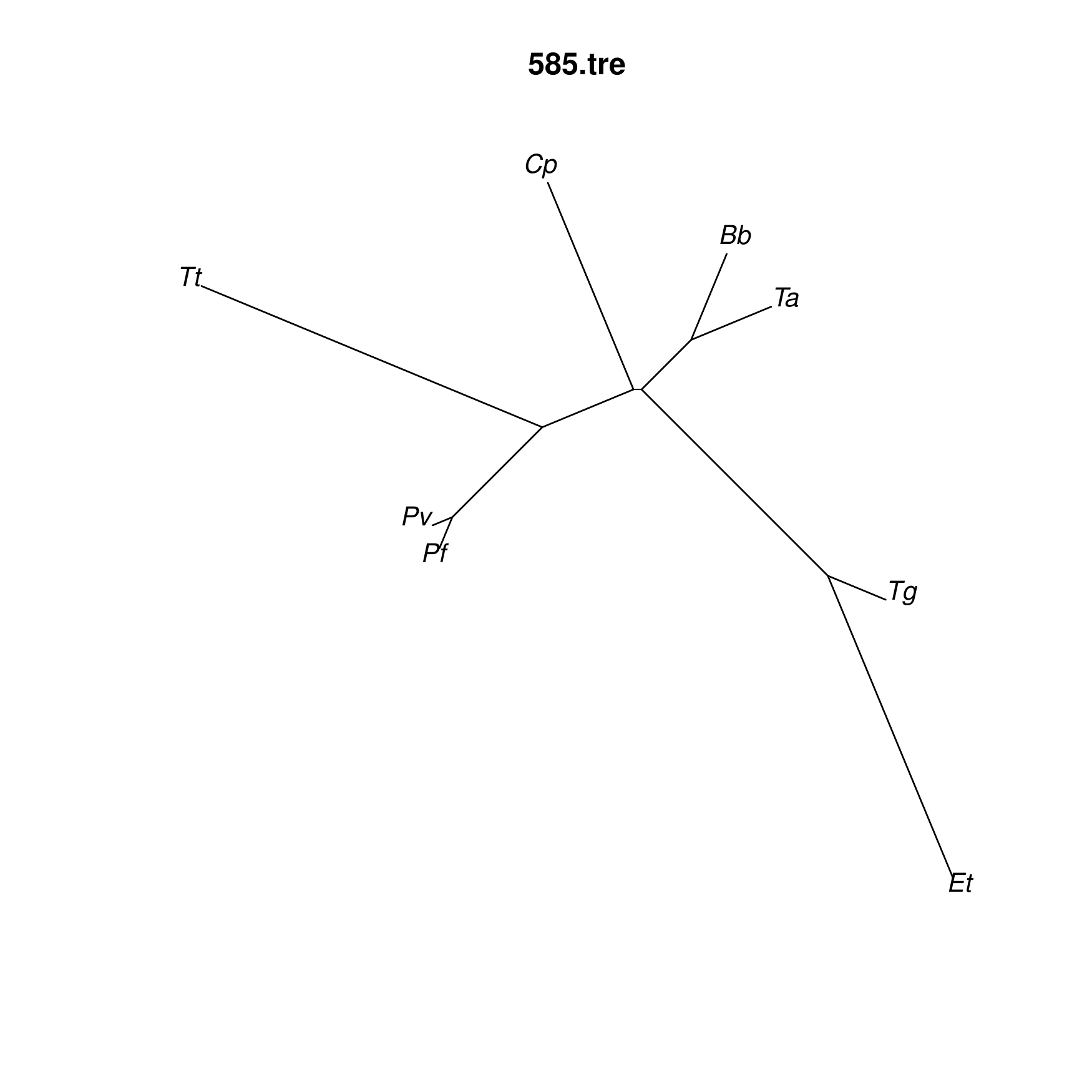}
\caption{A newly identified outlier from the Apicomplexa dataset.}

\end{figure}
\begin{figure}
  \centering
\includegraphics[width=\textwidth]{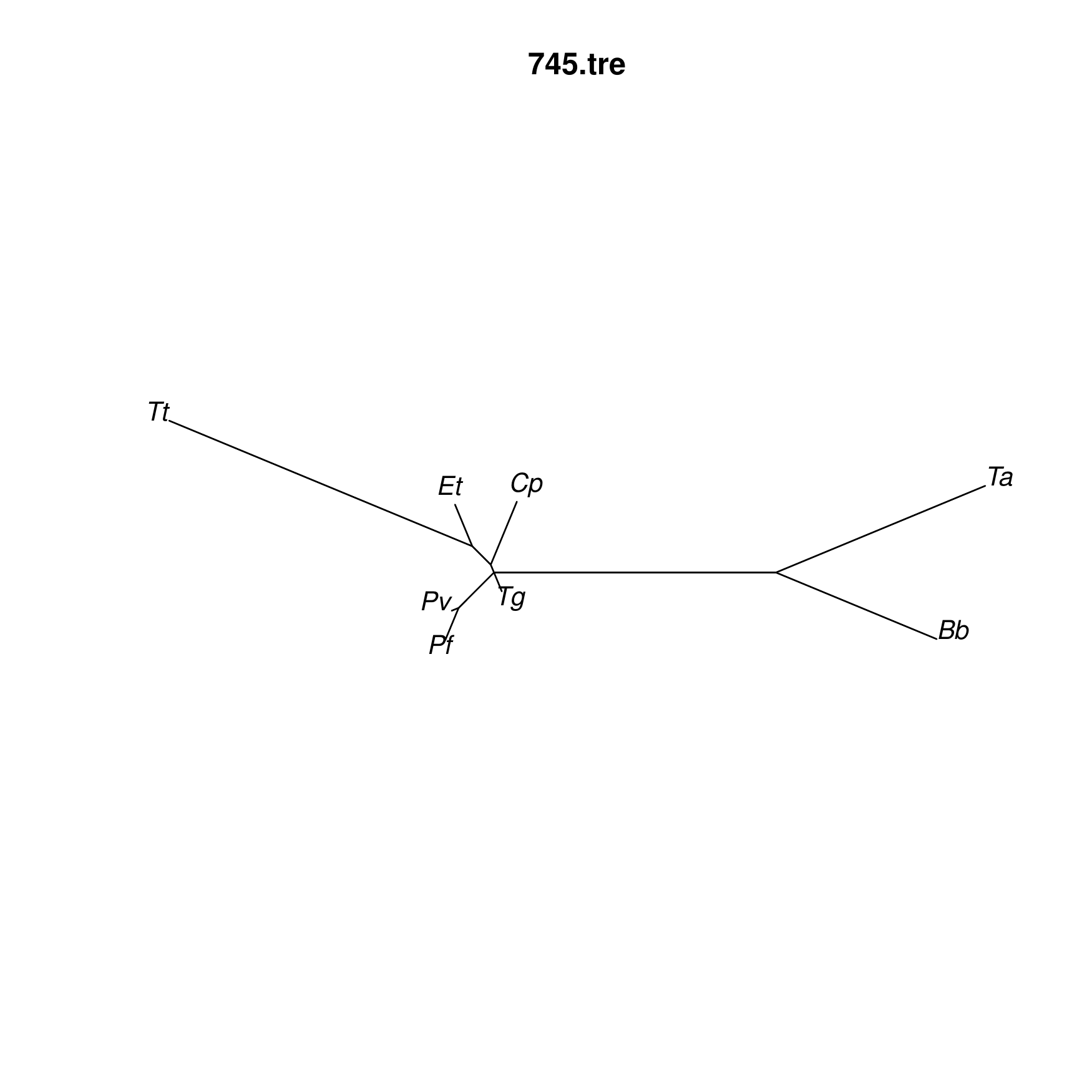}
\caption{A newly identified outlier from the Apicomplexa dataset.}

\end{figure}
\begin{figure}
  \centering
\includegraphics[width=\textwidth]{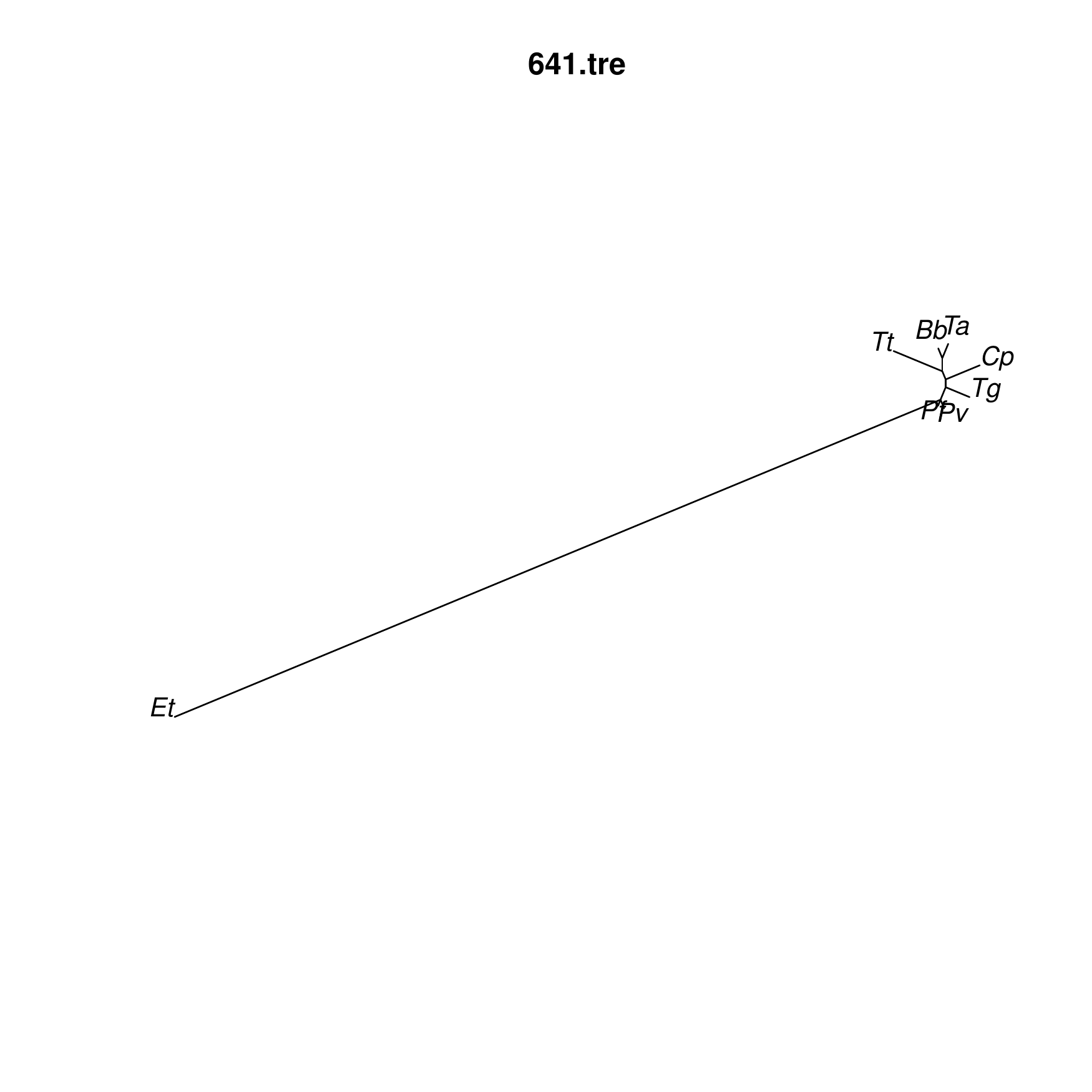}
\caption{A newly identified outlier from the Apicomplexa dataset.}

\end{figure}
\begin{figure}
  \centering
\includegraphics[width=\textwidth]{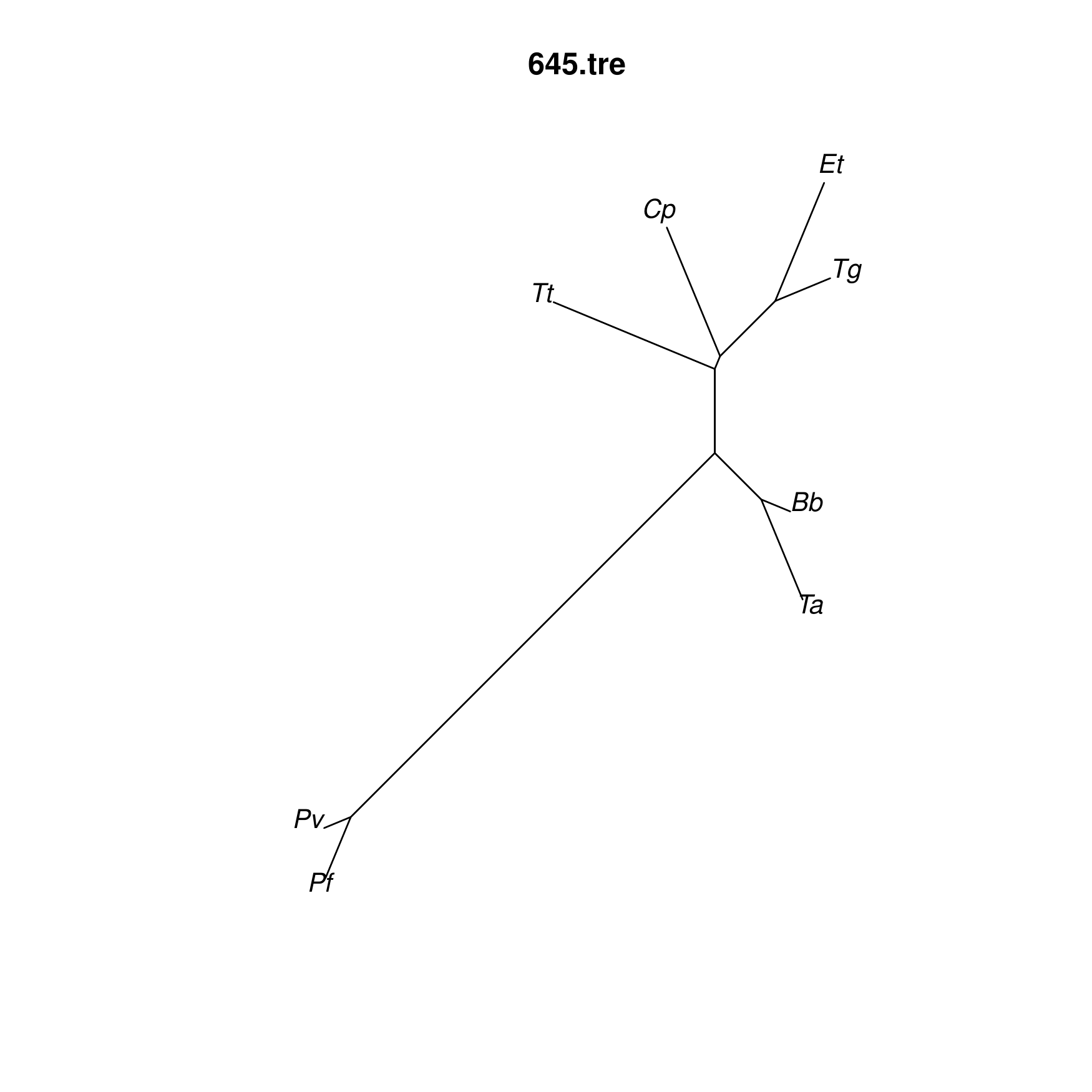}
\caption{A newly identified outlier from the Apicomplexa dataset.}

\end{figure}
\begin{figure}
  \centering
\includegraphics[width=\textwidth]{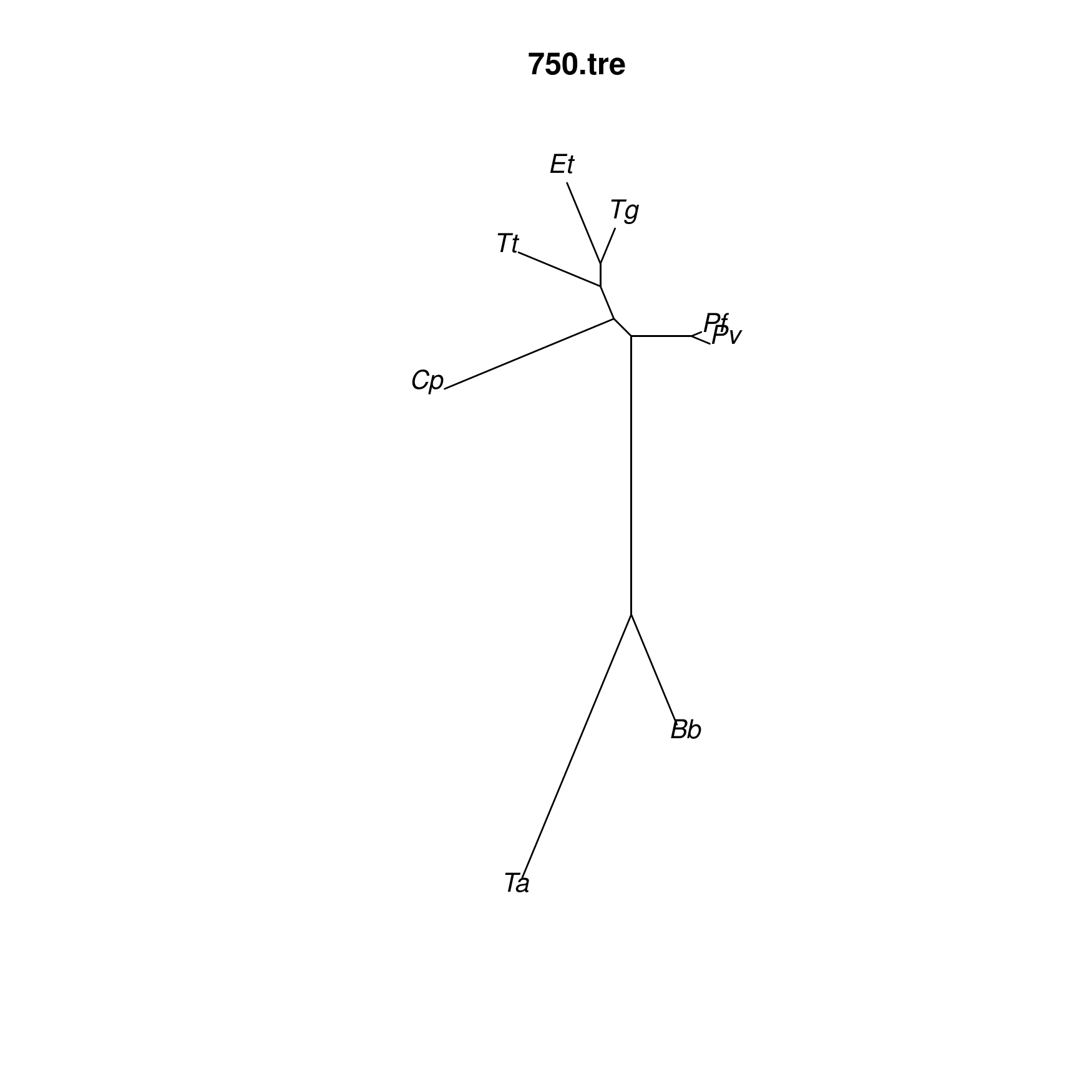}
\caption{A newly identified outlier from the Apicomplexa dataset.}

\end{figure}
\begin{figure}
  \centering
\includegraphics[width=\textwidth]{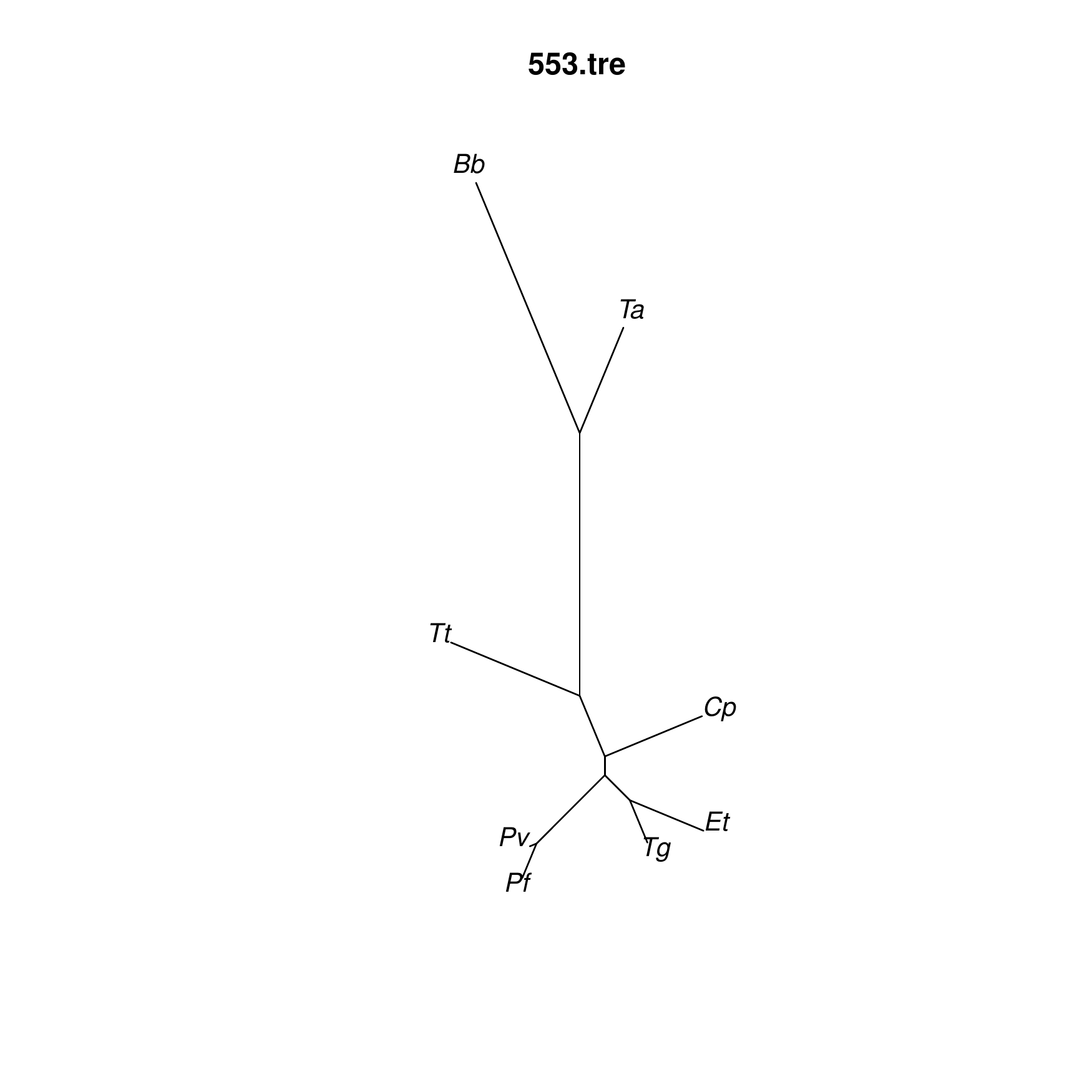}
\caption{A newly identified outlier from the Apicomplexa dataset.}

\end{figure}
\begin{figure}
  \centering
\includegraphics[width=\textwidth]{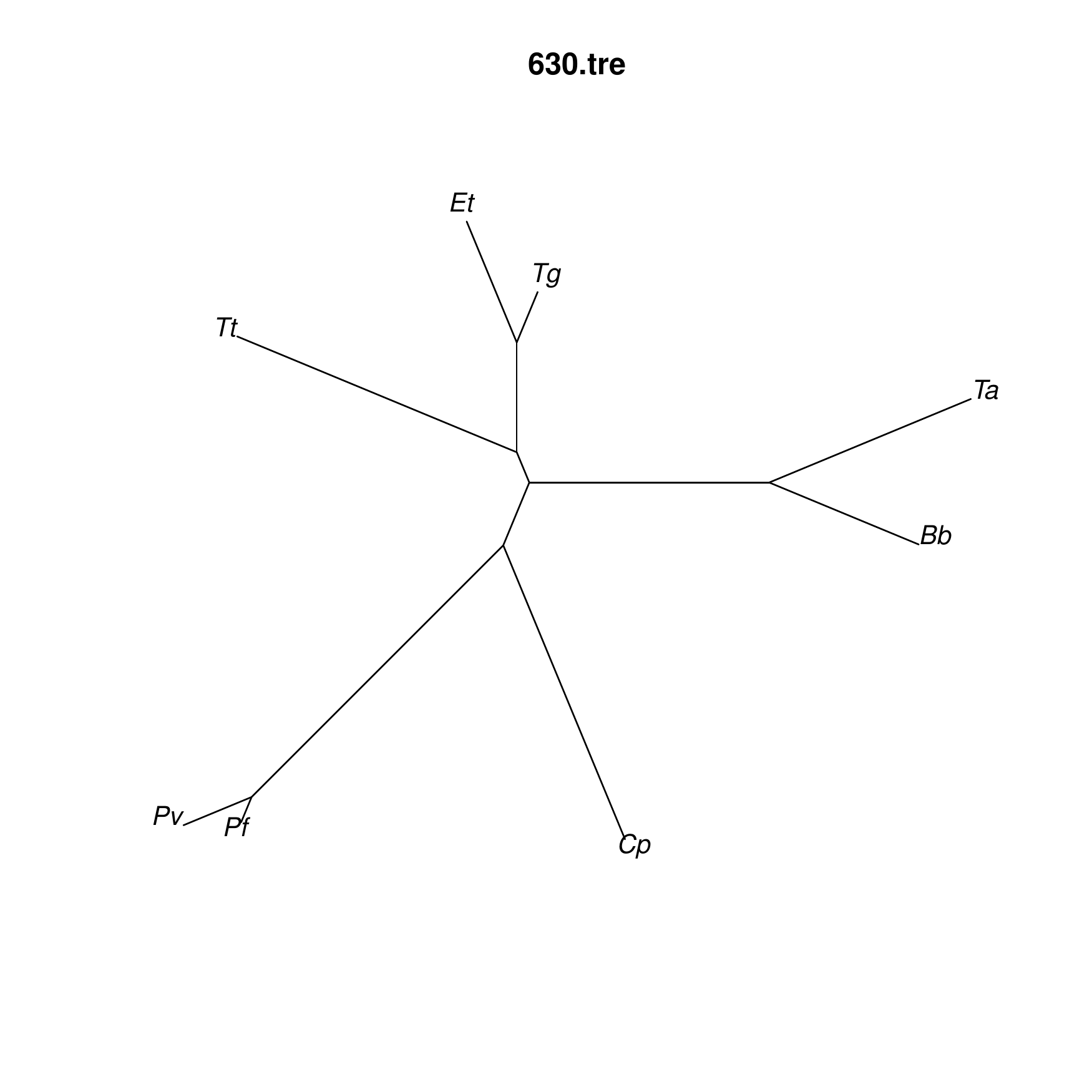}
\caption{A newly identified outlier from the Apicomplexa dataset.}

\end{figure}
\begin{figure}
  \centering
\includegraphics[width=\textwidth]{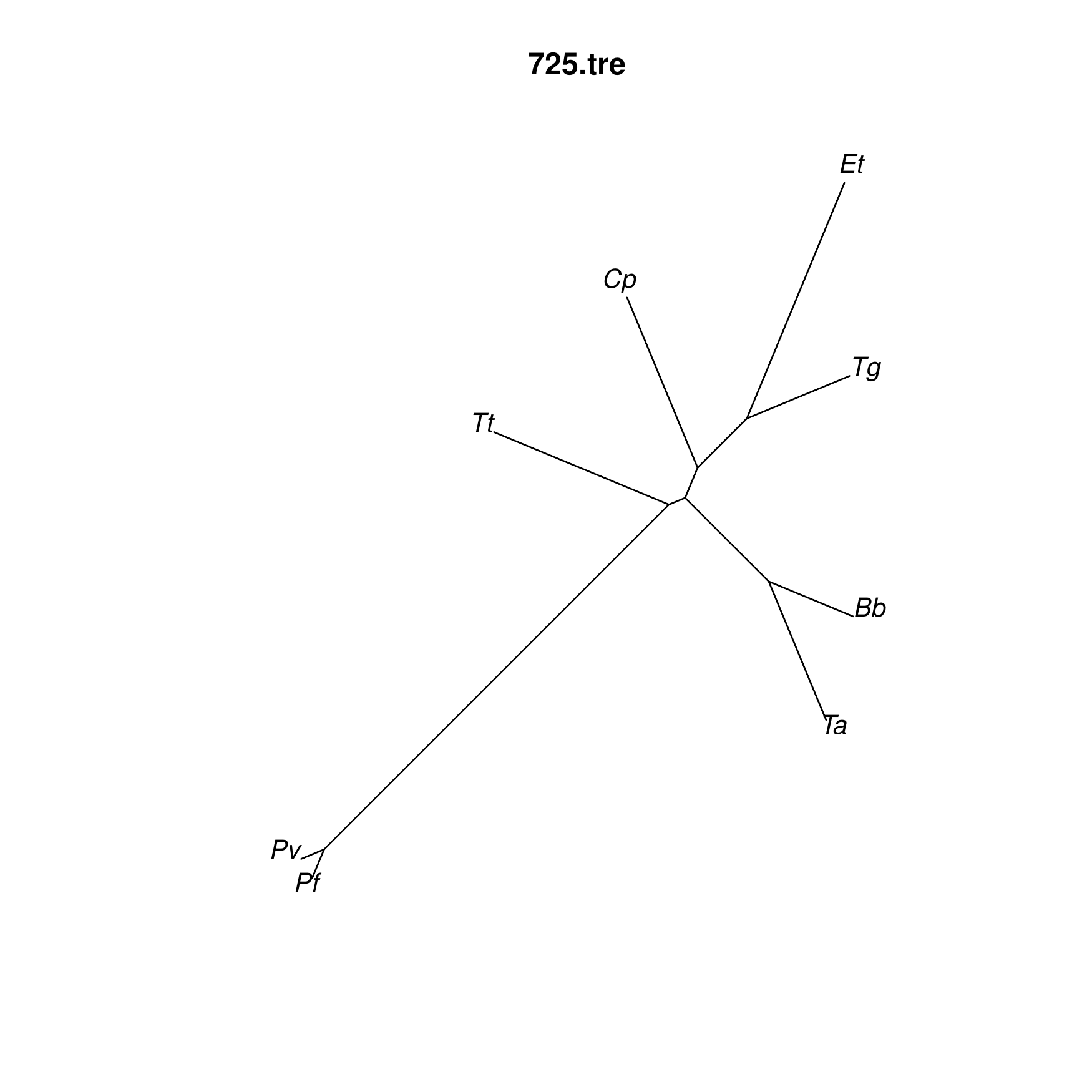}
\caption{A newly identified outlier from the Apicomplexa dataset.}

\end{figure}
\begin{figure}
  \centering
\includegraphics[width=\textwidth]{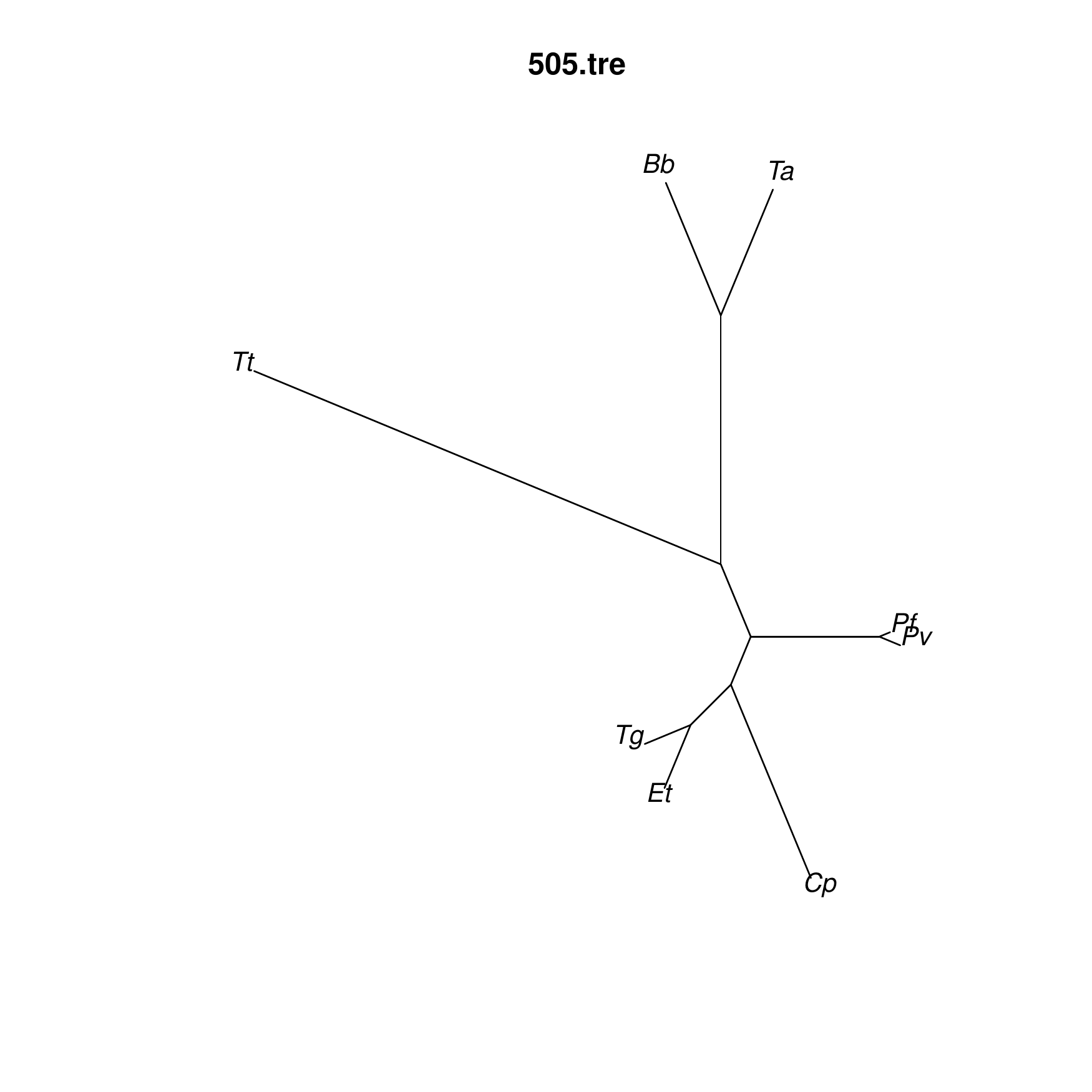}
\caption{A newly identified outlier from the Apicomplexa dataset.}
\end{figure}

\end{document}

%% file: new.tex
\section{Introduction}
\label{sec:const-intro}

One of the great opportunities offered by modern genomics is that
phylogenetics applied on a genomic scale (phylogenomics) should be
especially powerful for elucidating gene and genome evolution,
relationships among species and populations, and processes of
speciation and molecular evolution. However, a well-recognized hurdle
is the sheer volume of genomic data that can now be generated
relatively cheaply and quickly, but for which analytical tools are
lagging. There is a major need to explore new approaches to undertake
comparative genomic and phylogenomic studies much more rapidly and
robustly than existing tools allow.  Here, we focus on the problem of
{\em efficiently} identifying significant \emph{discordance} among a
set of gene trees, as well as estimating the distribution of gene
trees from the given set of trees.

The \kde{} algorithm introduced in \citet{weyenberg2014nonparametric}
is an adaptation of classical kernel density estimation to the metric
space of phylogenetic trees defined by \citet{billera2001geometry}. It
is a computationally efficient method of estimating the density of the
trees over the Billera-Holmes-Vogtman (BHV) treespace, and relies on a
fast implementation of the BHV geodesic distance function provided by
\citet{owen2011fast}. The method then uses the density estimates to
identify putative outlier observations. This paper describes an
improvement to \kde{}, whereby the kernel normalizing constants, are
estimated through the use of the novel holonomic gradient methods
\citep{koyama2014software,marumo2014properties}.

In our original paper describing the \kde{} method, we propose a
nonparametric estimator of the form,
\[\hat{f}(T) \propto \frac{1}{N} \sum_{i=1}^N k(T, T_i,h_i).\]
In the \kde{} software, the kernel function implemented is a
spherically symmetric gaussian kernel, i.e.
\begin{equation}
 k(T, T', h) \propto \exp \left( { -
    {\left({\frac{d(T,T')}{h}}\right)}^2 }\right).
\label{eq:kdef}
\end{equation}

Since we are, for the moment, interested primarily in using the
estimator $\hat f$ for outlier detection, knowledge of the overall
proportionality constant for $\hat f$ is not of significant
importance. However, it is important to know how the normalizing
constant associated with $k(T,T',h)$ varies with the selected
bandwidth and with the location of the kernel's center. In our
original paper we argued that, in practice, estimates of the
normalizing constant do not appear to have significant systematic
variation, and that assuming a constant value was a reasonable first
approximation. This paper presents basic results and techniques for
obtaining better approximate values for these normalizing
constants.

In the case case of Euclidean $k$-space with the usual metric, the
kernel (\ref{eq:kdef}) corresponds to a spherically symmetric
multivariate normal distribution centered on the point $T'$, and the
kernel normalization constant is given by
\begin{equation}
c(T',h_i) = (2\pi h_i)^{-k/2}.
\label{eq:normconst}
\end{equation}
Note that not only is there a simple closed form solution for the constant, but
the constant is invariant under changes to the central point $T'$.
However, when applied to the BHV 
treespace with
the geodesic metric, not only is such a closed form solution
apparently unavailable, but it is also clear that the constant will
depend on the location of the central point. 

For example, consider the case where $T' = 0$, i.e., the star tree, located
at the origin of BHV space. In this case, the kernel integral, 
\begin{equation}
  \label{eq:kernconst}
  c(T',h)=\int_{\mathcal T} k(T,T',h)\,dT,    
\end{equation}
is symmetric over each orthant comprising BHV treespace. Within each
orthant, the integral is equivalent to the normalizing constant of a
zero-mean multivariate normal truncated to $\mathbb R_n^+$. Thus, expression
(\ref{eq:kernconst}) is equivalent to the number of orthants in the
space, $n_O$, multiplied by the corresponding truncated normal constant.

On the other extreme, consider a central tree $T'$ such that every
edge is large compared to the bandwidth $h$, i.e. the tree is
relatively far away from any orthant boundary. In this case, the
kernel integral will be very close to the value given in expression
(\ref{eq:normconst}). If the central point is placed arbitrarily far
away from any orthant boundary, then the integral over any orthant
other than the one containing ${T'}$ can be made arbitrarily small.
Thus, the integral over the orthant containing ${T'}$ itself will be
an increasingly good estimate of the entire normalizing
constant as the central point is moved further away from orthant boundaries.

The updated version of \kde{} presented in this paper improves on the
first generation algorithm by estimating the kernel normalizing
constants~$c(T',h)$. This is accomplished by finding bounding
functions in each orthant which can be more easily integrated than the
true kernel function. While some analytic simplification is possible,
certain expressions cannot be evaluated other than through numerical
methods.

\subsection{Holonomic Gradient Method}
\label{sec:holom-grad-meth}
The \emph{holonomic gradient method} (HGM) is a non-stochastic
numerical method for calculating certain types of integrals. The HGM
is a variation on the gradient descent method of function
optimization, and is suitable for application to holonomic functions
\citep{nakayama2011holonomic,koyama2014software}. Roughly, a holonomic
function is a solution to a homogenous ordinary differential equation
with polynomial coefficients \citep{zeilberger1990holonomic}. Several
integrals of interest to statisticians turn out to be expressible as
solutions to an optimization problem within a holonomic system.

For our present problem two cases are of particular
use. \citet{marumo2014properties} demonstrates the use of HGM to
calculate the normalizing constant for a multivariate normal
distribution truncated to the positive orthant, i.e., $\R_n^+$. In addition,
\citet{hayakawa2014estimation} provides the constants for the
so-called exponential-polynomial family of probability densities,
\[ f(x|\theta_1,\ldots,\theta_k) \propto
\exp\left(x\theta_1+\ldots+x^k\theta_k\right).\]

As was briefly discussed in section \ref{sec:const-intro}, BHV
treespace is a simplical complex of positive Euclidean orthants, and
the normalizing constant for a truncated multivariate normal
distribution is an ingredient for a scheme to approximate the kernel
constants in BHV treespace. In section \ref{sec:results} we show that
it is possible to use the normalizing constants for the truncated
multivariate normal and the exponential-polynomial family, computed
either by HGM or some other method, to construct approximations to the
kernel normalizing constants for BHV space.

\section{Methods}
\label{sec:methods}
\subsection{Normalizing Constants}
In this paper we use $k$ to denote the \emph{unnormalized} kernel
function, i.e., with unit constant of proportionality in
(\ref{eq:kdef}). If we are given a fixed tree~$T_0$ and
bandwidth~$h$, our objective is to compute bounds for the
integral $K(T_0,h)=\int k(T,T_0,h)\; dT$ over the entire BHV treespace, so that
we may normalize the kernel function.

One suitable lower bound function is based on the use of
the triangle inequality.
\begin{lemma}
  For any pair of trees, $k(T,T',h) \ge \underline k(T,T',h).$
  Where,
  \begin{equation*}
    \underline k(T,T',h) = \exp\left(-\frac{(d(T,0)+d(0,T'))^2}{h^2}\right).
  \end{equation*}
  \label{lemma1}
\end{lemma}
\begin{proof}
  This is an immediate consequence of the fact that the geodesic path
  between any two trees is the shortest path connecting the trees. In
  particular, it is shorter than the cone path,
  $d(T,T')\le d(T,0)+d(0,T').$
\end{proof}

However, $\underline k$ does better than simply providing a global
lower bound for $k$. In fact, the bound is sharp, as $\underline k$ is
equivalent to $k$ whenever the geodesic between $T$ and $T'$ passes
through the origin. This turns out to be a quite common occurrence, as
geodesics between trees which are not separated by a small number of
NNI interchanges are likely to pass through the origin. As a result
for much of the space, integrating over $\underline k$ will be
equivalent to integrating over $k$ itself. Happily, integrating
$\underline k$ over a single orthant affords an opportunity for
analytical simplification.
\begin{thm}\label{thm:thm1}
  Let $O$ be an arbitrary fixed orthant in BHV treespace, and
  let $p$ denote its dimension. Then, the integral of $\underline
  k(T,T',h)$ over that orthant is given by the expression
  \begin{align}
    \underline C_O(T',h)  := &\int_O \underline k(T,T',h)\; dT \nonumber\\
                                = &
    \frac{\pi^{p/2}e^{-{d(0,T')^2/ h^{2} }}}
    {2^{p-1}\Gamma(p/2)} \underline A(T',h).
    \label{eq:Czero}
  \end{align}
  Where, if we let  $\theta_1=-2d(0,T')/h^2$ and $\theta_2=-h^{-2}$, then
  \begin{equation}
    \label{eq:Alower}
   \underline A(T',h) = \int_0^\infty r^{p-1}
  \exp\left(\theta_1r+\theta_2r^2 \right)\,dr.
  \end{equation}
\end{thm}
\begin{proof}
  The distance $d(T,0)$ is the usual $l_2$-norm of
  the vector of edge lengths for the tree $T$, and $O$ is the
  positive orthant~$\R_p^+$. Thus, if we express the
  integral over $O$ in an angular coordinate system, 
  \begin{align}
    \underline C_O(T',h)
    &= e^{-\frac{d(0,T')^2}{h^{2}}}
        \int_O
        e^{-\frac{(d(T,0)^2+2d(0,T')d(T,0))}{h^{2}}}\, dT\nonumber\\
    &= e^{-\frac{-d(0,T')^2}{h^{2}}} \int_0^\infty\int_\Theta
        e^{\theta_1r+\theta_2r^2} dV(r,\Theta)
    \label{eq:radialIntegral}
  \end{align}
  Now the volume element in $\R^p$ in an angular coordinate system is
  \[dV(r,\Theta) = r^{p-1} dr
  \prod_{k=1}^{p-1}\sin^{p-k-1}(\theta_k)\,d\theta_k,\]
  and it so happens that one of the definitions of the Beta function
  yields,
\[\int_0^{\pi/2} \sin^{p-k-1} (\theta) \,d\theta = \frac{1}{2}
B((p-k)/2,1/2).\]
Integrating in all the radial coordinates yields a product of beta
functions which telescopes down to the constant appearing in (\ref{eq:Czero}).
This reduces the problem to a single integral in the radial coordinate, which
is equivalent to $\underline A(T',h)$.
\end{proof}

Unfortunately, the function $\underline A(T',h)$ has no general closed
form solution. However, there are several methods that we can use
to obtain a numerical estimate of this value. The HGM method developed
in \citet{hayakawa2014estimation} is one such method for obtaining
this value. It is also reasonable to calculate this particular integral
using classical quadrature methods.

\begin{lemma}
  Following the notation of \citet{hayakawa2014estimation},
  \[ \underline A(T',h) = \partial_1^{p-1} \Atwo.\]
  Here, $\Atwo$ is the normalizing constant for the
  exponential-polynomial distribution of order 2, the $\theta$ are
  defined as in Theorem \ref{thm:thm1}, and $\partial_1^m$ means the
  $m$-th partial derivative with respect to $\theta_1$.  Furthermore,
  Hayakawa gives the following equivalence for the first partial derivative
  \[ \partial_1\Atwo = -\frac{1}{2\theta_2}\left\{1 +
    \theta_1\Atwo\right\}, \]
  and for the the higher derivatives, $m\ge 2$, the partials can be expressed
  recursively in terms of lower order derivatives,
  \begin{align}
  \partial_1^{m}\Atwo=-\frac{1}{2\theta_2}\{&
                         (m-1) \partial_1^{m-2}\Atwo+\nonumber\\ 
                     &\theta_1\partial_1^{m-1}\Atwo\}.\label{eq:partialA}
  \end{align}
\end{lemma}
\begin{proof}
  See \citet{hayakawa2014estimation}, Section 2, equations (4) and
  (7). The latter expression can be easily obtained by differentiation
  of the expression for the first partial.
\end{proof}
These results are sufficient to use the \emph{hgm} package
described in \citet{koyama2014software} to implement the lower bound
for the orthant integral,~$\underline C_O(T',h)$. 

The desired normalization constant for function (\ref{eq:kdef}) can be
decomposed as the sum of integrals over each orthant in BHV
space. Thus, if $n_O$ is the number of orthants in the BHV space, then
$n_O \cdot \underline C_O(T',h)$ is a crude lower bound for the
overall normalizing constant. Although this is a poor bound, it can be
improved by obtaining better bounds for the contribution from various
orthants and adjusting accordingly.

The most obvious orthant to begin with is the orthant containing the
``central'' tree $T'$, which we shall call $O_{T'}$. This is the
orthant where the difference between $k$ and $\underline k$ will be
the greatest, and thus the largest improvement to the bounding
constant is to be found here. Note that in this case, the integral
over $O_{T'}$ is given by,
\begin{align}
C_{O_{T'}}(T',h) &= \int_{O_{T'}}
                     \exp\left(d(T',T)^2/h^2\right)\,dT\nonumber\\
  &=\int_{\R_p^+}\exp\left(-||x-x_{T'}||/h^2\right)\,dx.
\end{align}
This is simply the integral of a radially-symmetric multivariate
gaussian kernel centered at the point $T$ over the positive
orthant. Such a normalizing constant can also be calculated using HGM,
and an implementation is included in the \emph{hgm} R package
\citep{koyama2014software}.

Further improvements to the integral for orthants which adjoin directly to
$O_{T'}$ can be made by noting that a relationship similar to that of
Lemma \ref{lemma1} will hold, but with the third point in the triangle
inequality being somewhere on the orthant boundary, instead of the
origin. However, in practice the improvements to the bounds obtained
in this way are quite small, given the typical values of the
bandwidths which occur in practice and the small number of orthants to
which the calculations apply. For this reason, and in the interests of
controlling the overall numerical complexity of the \kde{} algorithm,
these ``second-order'' approximations are not implemented at
at this time, but may appear in future updates. 

\subsection{Outlier Test}
In \cite{weyenberg2014nonparametric} we chose to implement a outlier
test of the form, $\hat f(T_i) < c^*$, where the critical value $c^*$
is selected using Tukey's quartile method,
\[ c^* = Q_1 - k(Q_3-Q_1). \]
Here, $\hat f(T)$ denotes our density estimate for tree $T$, and the
quartiles,~$Q_1,Q_3$,are calculated using all observed tree density estimates.

Further experimentation with the method has suggested that better
performance is obtained if the tree density scores are transformed to
the log-scale before the classification step takes place,
$\log f(T_i) < c^*$. This transformation was chosen because the raw
scores,~$f(T_i)$, are themselves bounded below by zero, and the log
transformation removes this bound. The quantiles used to compute the
critical value are also obtained using the log-transformed scores. Due
to the better performance characteristics of this method, the default
classifier algorithm for \kde{} has been changed to operate on the
log-density scale.

\subsection{Leaf Edges}
The dimension of the orthants comprising tree space is determined by
the number of taxa in the trees, with each edge in the fully-resolved
tree contributing a dimension to each of the orthants comprising the
space. However, the $n$ leaf edges are represented in the space in
such a way that the space can be decomposed into a Cartesian product
$S\times \R_n^+$. The portion of the space $S$ is associated with the
internal edges of the trees, while the positive euclidean orthants
$\R_n^+$ are associated with the leaf edges
\citep{billera2001geometry}. Because of this decomposition, there is
not a large amount of topological information contained in the portion
of the space corresponding to the leaf edge lengths. 

If we remove the leaf edges from the calculation, the dimension of
treespace can be reduced by approximately half, while retaining the
important topological information. This has the benefit of simplifying
the overall density estimation problem, as well as the computation of
the normalizing constant estimates by the HGM methods. While the
original \kde{} algorithm included the leaf edges in the geodesic
calculations, the updated version omits them from consideration. A new
option flag allows for restoration of the original functionality.



\section{Results}
\label{sec:results}
The updated version of the \kde{} software package is available both
from CRAN (the official R package system), as well as from the
official development repository on Github. (\href{github.com/grady/kdetrees}{github.com/grady/kdetrees}).

\subsection{Simulations}
A set of simulated datasets were constructed and analyzed, using a
similar design as the first simulation described in
\citet{weyenberg2014nonparametric}. The simulations measure the true
and false positive rates for identification of known outlier trees
within a set of trees drawn from a common distribution. Results of the
simulation are summarized as ROC curves, and are presented in Figure
\ref{fig:hgmroc}.

\begin{figure}
  \centering
  \includegraphics[width=\columnwidth]{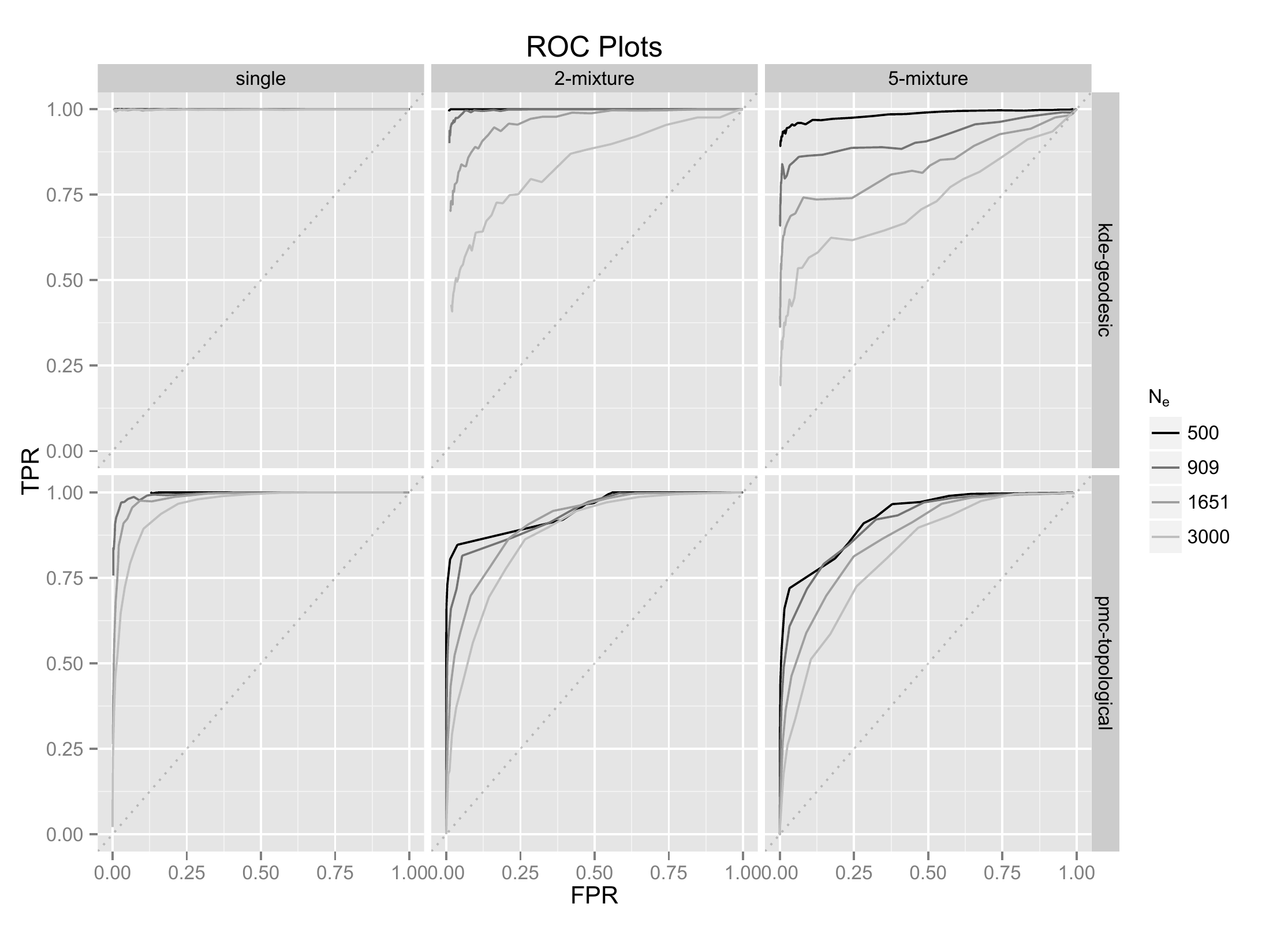}
  \caption[Updated \kde{} ROC curves.]{Receiver operating
    characteristic (ROC) plots comparing the updated \kde{}
    algorithm with Phylo-MCOA. ROC plots summarize the true and false
    positive rates for a binary classifier as the tuning parameters
    are changed. A perfect classifier would be represented as a single
    point at the upper left corner of the plot, while a completely
    random scheme would follow the dotted diagonal lines.  Curves are
    shown from simulated coalescent data, using a variety of effective
    population sizes ($N_e$). Larger values for $N_e$ correspond to
    more variability in the generated trees, and thus a more difficult
    classification problem.}
  \label{fig:hgmroc}
\end{figure}

\subsection{Apicomplexa}
The Apicomplexa data set presented by \cite{kissinger} consists of
trees reconstructed from 268 single-copy genes from the following
species: {\it Babesia bovis} (Bb) \citep{Brayton2007} (GenBank
accession numbers AAXT01000001--AAXT01000013), {\it Cryptosporidium
  parvum} (Cp) \citep{Abrahamsen2004} from CryptoDB.org
\citep{Heiges2006}, {\it Eimeria tenella} (Et) from GeneDB.org
\citep{Hertz-Fowler2004}, {\it Plasmodium falciparum} (Pf)
\citep{Gardner2002} and {\it Plasmodium vivax} (Pv) from PlasmoDB.org
\citep{Bahl2003}, {\it Theileria annulata} (Ta) \citep{Pain2005} from
GeneDB.org \citep{Hertz-Fowler2004}, and {\it Toxoplasma gondii} (Tg)
from Toxo-DB.org \citep{Gajria2008}. A free-living ciliate, {\it
  Tetrahymena thermophila} (Tt) \citep{Eisen2006}, was used as the
outgroup. To this set of sequences, we appended the Set8 gene, which
has been identified by \cite{Kishore} as a probable case of horizontal
gene transfer from a higher eukaryote to an ancestor of the
Apicomplexa.  This is the same data set analyzed as part of the
original \kde{} paper, which was analyzed again with the updated
algorithm. The new set of outlier trees is presented in Table
\ref{apicomplexa_names}. The newly identified set of outlier trees are
presented in a series of figures at the end of this \sect{}. The
figures are in ascending order by the kdetrees tree score, i.e., the first
tree depicted is the furthest outlying tree.

\begin{table*}
\centering
\caption[Outlier trees identified by the new kdetrees algorithm.]{Apicomplexa
  gene sets identified as outliers by the updated \kde{}. Genes which
  were not identified as outliers by the original algorithm are marked
  with a $*$.} 
\label{apicomplexa_names} 
\vspace{0.5\baselineskip}
{\footnotesize`
  \begin{tabular}{rlp{0.15\textwidth}p{0.6\textwidth}}
    \hline
    No.$^a$ & GeneID$^b$ & Functional Annotation & Analysis\\\hline
    $*$ 472  &PF14\_0059 & hypothetical protein & Tree topology inconsistent with phylogeny.  Bb and Cp on same branch, with Ta distant from sister species Bb.  Sequence alignment looks good in some regions, but with numerous gaps and other regions with poor alignment.  Multiple homopolymer stretches in Pv and Pf.\\
    $*$ 478  &PF14\_0326 & hypothetical protein & Tree topology not consistent with phylogeny of the species.  Bb branches with the outgroup Tt instead of it’s closely-related sister species Ta.  Poor alignment with numerous gaps, numerous homopolymer stretches, particularly in Et.\\
    488  &PF08\_0086 & RNA-binding protein, putative &Significant sequence length disparity (164 a.a. for Ta vs 1075a.a. for Pf).  Generally good sequence alignment in one region of ~100 residues; otherwise, alignment is poor.\\
    $*$ 505  & PF14\_0143& protein kinase, putative & Ta/Bb and Pf/Pv not monophyletic; split by outgroup Tt.  Good sequence alignment in multiple blocks, but significant sequence length differences.  Pf/Pv have multiple insertions and Et and Cp sequences are truncated.\\
    515  &PFA0390w&  DNA repair exonuclease, putative & Short  sequences  for Et and  Cp.   Several  homopolymer stretches  in Et.  Modest to good alignment in multiple blocks, Et being an exception in several  regions.  Possible incorrect annotation of Et sequence.\\
    $*$ 553 &PFC0730w &  conserved protein, putative & Tree topology inconsistent with phylogeny.  Bb and Ta are distant not monophyletic with Pv/Pf.  Short regions exhibiting good sequence alignment.  Et sequence is truncated.\\
    $*$ 578  & PF14\_0042& U3 small nucleolar ribonucleoprotein, U3
    snoRNP putative & Tree topology very inconsistent with phylogeny; Tg branch with outgroup Tt, Et branch with Bb and Ta.  Poor alignment.  Significant sequence length differences; Tg sequence is 4126 residues in length, Cp and Tt ~2000 residues, Pf and Pv are ~450 residues.\\
    $*$ 585 & PF10\_0054 & hypothetical protein & Cp exhibits anomalous placement in tree.  Significant sequence length differences; Pf, Pv, Tg about 1100 residues, Et only 349 residues, so numerous gaps in alignment.  Some regions show good alignment.\\
    $*$ 588  & PFI1020c &Inosine-5'-monophosphate dehydrogenase & Sequence alignment looks reasonably good.  Tree shows Cp branching with outgroup Tt and distant from other Api species.  Bb split from Ta.\\
    $*$ 630 & PFL2120w& hypothetical protein, conserved & Tree topology inconsistent with phylogeny.  Cp branching with Pv/Pf.  Ta/Bb not monophyletic with Pv/Pf.  Several blocks of sequence showing good alignment, but numerous gaps, due mostly to Ta insertions.\\
    641  & PFE0750c& hypothetical protein, conserved & Et on a very long branch with other species tightly clustered.  Large difference in sequence lengths; 269 residues for Et vs. 848 for Pf.  Central region with modest to good alignment; Et exhibited poor sequence identity. \\
    $*$ 645 & PF14\_0635& RNA binding protein, putative & Tree topology looks proper, although Pv and Pf are on a somewhat long branch.  Modest to good alignment.\\
    $*$ 662 & PF11\_0463 & coat protein, gamma subunit, putative & Multiple homopolymer stretches in Et sequence. Generally good alignment for all but Et; sequence might not be homologous.\\
    $*$ 725  & PF14\_0428& histidine -- tRNA ligase & Tree topology appears proper, but Pf/Pv on long branch.  Good alignment in two large blocks, but significant gaps and poor alignment in other regions.  Et sequence truncated (339 vs. ~1000 residues for others).  Ta sequence also truncated (583 residues).\\
    $*$ 745  & PF11\_0049 & hypothetical protein, conserved & Ta and Bb branch is distant from other Api species, which cluster tightly.  Regions of good sequence alignment by with several large gaps.  Sequence length differences; Pf and Pv = 3300 residues, Tg and Cp = 2600, Et only 347.\\
    $*$ 750  & PFE1050w & adenosylhomocysteinase
    (S-adenosyl-L-homocystein e hydrolase) & Ta and Bb somewhat
    distant from other Api species.  Relatively good sequence
    alignment, although Et sequence truncated (291 vs. ~480 for
    others).\\
\hline
\end{tabular}
}
\begin{flushleft}$^a$Based on geneset designations in \cite{kissinger}.\\
  $^b$Geneset represented by GeneID for \emph{Plasmodium
    falciparum}.\\
Pf = {\it Plasmodium falciparum}, Pv  = {\it Plasmodium vivax},  Bb = {\it Babesia  bovis},  Ta  = {\it Theileria annulata}, Et = {\it Eimeria tenella}, Tg = {\it Toxoplasma gondii},  Cp = {\it Cryptosporidium parvum}, and  Tt = {\it Tetrahymena thermophila} (outgroup).
\end{flushleft} 
\end{table*}

\section{Discussion}
\label{sec:discussion}
\subsection{Apicomplexa}
The outliers identified by \kde{} in the Apicomplexa dataset are
substantially different than those reported in our original paper. In
the original paper, many of the outliers were trees containing a
single edge much longer than any other edge. This was largely
attributable to one or more poorly aligned sequences within the larger
multiple sequence alignment. Thus, the disproportionate edges were
often leaf edges, and as a result of the changes in the algorithm, the
leaf edges are no longer taken into consideration by the default
settings. As a result, many of the trees identified as outliers in the
original paper are no longer identified as outliers by the default
parameter settings of the improved version. Of course, there is a new
option flag available, which can be used to restore the behavior
found in the original paper. The software also retains the
``dissimilarity mode'' from the original implementation, which always
uses the terminal branch length information.

The cumulative result of the changes in the algorithm is an increased
focus on differences in topology in the dataset. The new set of 16
outlying trees differ from the non-outliers primarily in the placement
of the Bb, Cp, Ta, and Tt genes within the trees. In the non-outlier
trees, Cp generally forms a clade with Tt, while Ta forms a clade with
Bb. In the outlier trees, however, these taxa are placed in widely
varying locations within the trees, as demonstrated by the tree
drawings of outlier trees found at the end of this \sect{}.

Together, these results demonstrate that the updated \kde{} algorithm
is more sensitive to topological differences in the trees than the
previous version, at the expense of the loss of information from the
terminal edge lengths. However, the original functionality using the
terminal edge information is still available for use, by setting the
appropriate flags.


\subsection{Simulations}
The performance of the classifier with the simulated datasets is
substantially better than the original version of the
algorithm. Although there is a modest performance penalty associated
with the estimation of the normalizing constants, \kdet{} remains
significantly faster than the competitor \pmc{} algorithm \citep{MCOA}.

When the non-outlier trees are drawn from a single coalescent
distribution, the performance of the classifier is nearly perfect,
identifying the correct outlier in every simulation iteration, even
when the variance of the coalescent distributions (controlled by the
effective population size parameter~$N_e$) was quite large. (Of
course, due to the nature of the classifier, false positives are
inevitable if the tuning parameter is chosen poorly.) In the more
difficult cases where the non-outliers were drawn from a mixture of
coalescent distributions, the updated algorithm remained superior to
the Phylo-MCOA algorithm, showing greater area under the ROC curve for
all cases except for the case of the most highly variable 5-part
mixture distribution for the non-outlier trees.

\section{Conclusion}

 Our proposed method is motivated by the fact that existing methods
 of phylogenetic analysis and tree comparisons are not adequate for
 genomic scale phylogenetic analysis, particularly in cases of certain
 non-canonical evolutionary phenomena.
Furthermore, the scenario in our mixed coalescent distribution
simulation---where the non-outlier trees are sampled from an unknown
 mixture of distributions---cannot be handled by parametric methods,
 with the possible exception of the genome spectral methods.  However,
 even the genome spectral methods ignore possible statistical
 dependencies between different feature spectra.  In contrast, we
 propose analyzing a collection of gene trees without reducing gene
 trees to summarizing information.  Our \kdet{} approach also possesses
 a considerable advantage in speed over other methods, which is of
 paramount importance for a tool used in whole-genome phylogenetic
 analysis.

 In addition, one of the applications of our method is an inference on the
 species tree or a tree that reflects the evolution of most genes in
 the genome.  We can use our method to identify genes which produce
 discordant trees (outlying trees) and then we can remove them from
 phylogenetic analysis.  By doing this we can use the genes that share
 the same evolutionary history and we can build a tree that reflects
 the evolution of the species or that of most of the genome.

We have been interested in developing a phylogenomic pipeline that is
convenient and accessible, as well as robust. To accomplish this aim, an
important problem that needs attention is to comparing
thousands of gene phylogenies across whole genomes.  Thus, our
approach is focused on the 
efficiency of the algorithm in terms of computational complexity and
memory requirements, with less emphasis on achieving the highest
classification accuracy possible.
Thus, it is very important for us to improve the accuracy of the
method while we maintain the speed of the algorithm.  
Compared with the original \kdet{} and the Phylo-MCOA algorithm, in
this paper, we have demonstrated the 
significant improvement in the accuracy of the method without costing
the computational speed (see 
the simulation results in Figure \ref{fig:hgmroc}). 

In future work we intend to apply \kdet{} to clustering trees
based on similarity.  
Unsupervised clustering is an important method to learn the structure
of unlabeled data. The aim of clustering methods is to group patterns
on the basis of a similarity (or dissimilarity) criteria where groups
(or clusters) are set of similar patterns. 

Many traditional clustering algorithms (e.g., K-Means, Fuzzy c-Means,
SOM and Neural Gas) have the version of kernel-based algorithms
\cite[]{Hur2000,Hur2001,Camastra2005}. The use of kernels allows to map
implicitly data into an high dimensional space, called feature space;
computing a linear partitioning in this feature space results in a
nonlinear separation between clusters in the input space. Using
kernels defined in the space of trees, mapped into vector spaces, to
evaluate trees' structural similarity allows us to use Kernel-based
clustering methods to group trees in the space.

%% file: wrapper.bbl
\begin{thebibliography}{}

\bibitem[\protect\citeauthoryear{Abrahamsen \bgroup \em et al.\egroup
  }{2004}]{Abrahamsen2004}
Mitchell~S. Abrahamsen, Thomas~J. Templeton, Shinichiro Enomoto, Juan~E.
  Abrahante, Guan Zhu, Cheryl~A. Lancto, Mingqi Deng, Chang Liu, Giovanni
  Widmer, Saul Tzipori, Gregory~A. Buck, Ping Xu, Alan~T. Bankier, Paul~H.
  Dear, Bernard~A. Konfortov, Helen~F. Spriggs, Lakshminarayan Iyer, Vivek
  Anantharaman, L.~Aravind, and Vivek Kapur.
\newblock Complete genome sequence of the apicomplexan, cryptosporidium parvum.
\newblock {\em Science}, 304:441--445, 2004.

\bibitem[\protect\citeauthoryear{Bahl \bgroup \em et al.\egroup
  }{2003}]{Bahl2003}
Amit Bahl, Brian Brunk, Jonathan Crabtree, Martin~J. Fraunholz, Bindu Gajria,
  Gregory~R. Grant, Hagai Ginsburg, Dinesh Gupta, Jessica~C. Kissinger, Philip
  Labo, Li~Li, Matthew~D. Mailman, Arthur~J. Milgram, David~S. Pearson,
  David~S. Roos, Jonathan Schug, Christian~J. Stoeckert, and Patricia Whetzel.
\newblock Plasmodb: the plasmodium genome resource. a database integrating
  experimental and computational data.
\newblock {\em Nucleic Acids Res.}, 31:212--215, 2003.

\bibitem[\protect\citeauthoryear{Billera \bgroup \em et al.\egroup
  }{2001}]{billera2001geometry}
L.J. Billera, S.P. Holmes, and K.~Vogtmann.
\newblock Geometry of the space of phylogenetic trees.
\newblock {\em Adv Appl Math}, 27(4):733--767, 2001.

\bibitem[\protect\citeauthoryear{Brayton \bgroup \em et al.\egroup
  }{2007}]{Brayton2007}
Kelly~A. Brayton, Audrey O.~T. Lau, David~R. Herndon, Linda Hannick, Lowell~S.
  Kappmeyer, Shawn~J. Berens, Shelby~L. Bidwell, Wendy~C. Brown, Jonathan
  Crabtree, Doug Fadrosh, Tamara Feldblum, Heather~A. Forberger, Brian~J. Haas,
  Jeanne~M. Howell, Hoda Khouri, Hean Koo, David~J. Mann, Junzo Norimine,
  Ian~T. Paulsen, Diana Radune, Qinghu Ren, Roger K.~Smith Jr., Carlos~E.
  Suarez, Owen White, Jennifer~R. Wortman, Donald P.~Knowles Jr.1, Terry~F.
  McElwain, and Vishvanath~M. Nene.
\newblock Genome sequence of babesia bovis and comparative analysis of
  apicomplexan hemoprotozoa.
\newblock {\em PLoS Pathog}, 3:e148, 2007.

\bibitem[\protect\citeauthoryear{Camastra and Verri}{2005}]{Camastra2005}
F.~Camastra and A.~Verri.
\newblock A novel kernel method for clustering.
\newblock {\em IEEE Transactions on Pattern Analysis and Machine Intelligence},
  27(5):801--804, 2005.

\bibitem[\protect\citeauthoryear{de Vienne \bgroup \em et al.\egroup
  }{2012}]{MCOA}
D.~M. de~Vienne, S.~Ollier, and G.~Aguileta.
\newblock Phylo-mcoa: A fast and efficient method to detect outlier genes and
  species in phylogenomics using multiple co-inertia analysis.
\newblock {\em Mol Biol Evol}, 2012.

\bibitem[\protect\citeauthoryear{Eisen \bgroup \em et al.\egroup
  }{2006}]{Eisen2006}
Jonathan~A Eisen, Robert~S Coyne, Martin Wu, Dongying Wu, Mathangi Thiagarajan,
  Jennifer~R Wortman, Jonathan~H Badger, Qinghu Ren, Paolo Amedeo, Kristie~M
  Jones, Luke~J Tallon, Arthur~L Delcher, Steven~L Salzberg, Joana~C Silva,
  Brian~J Haas, William~H Majoros, Maryam Farzad, Jane~M Carlton, Roger~K
  Smith, Jr., Jyoti Garg, Ronald~E Pearlman, Kathleen~M Karrer, Lei Sun, Gerard
  Manning, Nels~C Elde, Aaron~P Turkewitz, David~J Asai, David~E Wilkes, Yufeng
  Wang, Hong Cai, Kathleen Collins, B.~Andrew Stewart, Suzanne~R Lee, Katarzyna
  Wilamowska, Zasha Weinberg, Walter~L Ruzzo, Dorota Wloga, Jacek Gaertig,
  Joseph Frankel, Che-Chia Tsao, Martin~A Gorovsky, Patrick~J Keeling, Ross~F
  Waller, Nicola~J Patron, J.~Michael Cherry, Nicholas~A Stover, Cynthia~J
  Krieger, Christina del Toro, Hilary~F Ryder, Sondra~C Williamson, Rebecca~A
  Barbeau, Eileen~P Hamilton, and Eduardo Orias.
\newblock Macronuclear genome sequence of the ciliate \emph{Tetrahymena
  thermophila}, a model eukaryote.
\newblock {\em PLoS Biol.}, 4:1620--1642, 2006.

\bibitem[\protect\citeauthoryear{Gajria \bgroup \em et al.\egroup
  }{2008}]{Gajria2008}
Bindu Gajria, Amit Bahl, John Brestelli, Jennifer Dommer, Steve Fischer, Xin
  Gao, Mark Heiges, John Iodice, Jessica~C. Kissinger, Aaron~J. Mackey,
  Deborah~F. Pinney, David~S. Roos, Christian~J. Stoeckert, Haiming Wang, and
  Brian~P. Brunk.
\newblock Toxodb: an integrated toxoplasma gondii database resource.
\newblock {\em Nucleic Acids Res.}, 36:D553--D556, 2008.

\bibitem[\protect\citeauthoryear{Gardner \bgroup \em et al.\egroup
  }{2002}]{Gardner2002}
Malcolm~J. Gardner, Neil Hall, Eula Fung, Owen White, Matthew Berriman,
  Richard~W. Hyman, Jane~M. Carlton, Arnab Pain, Karen~E. Nelson, Sharen
  Bowman, Ian~T. Paulsen, Keith James, Jonathan~A. Eisen, Kim Rutherford,
  Steven~L. Salzberg, Alister Craig, Sue Kyes, Man-Suen Chan, Vishvanath Nene,
  Shamira~J. Shallom, Bernard Suh, Jeremy Peterson, Sam Angiuoli, Mihaela
  Pertea, Jonathan Allen, Jeremy Selengut, Daniel Haft, Michael~W. Mather,
  Akhil~B. Vaidya, David M.~A. Martin, Alan~H. Fairlamb, Martin~J. Fraunholz,
  David~S. Roos, Stuart~A. Ralph, Geoffrey~I. McFadden, Leda~M. Cummings,
  G.~Mani Subramanian, Chris Mungall, J.~Craig Venter, Daniel~J. Carucci,
  Stephen~L. Hoffman, Chris Newbold, Ronald~W. Davis, Claire~M. Fraser, and
  Bart Barrell.
\newblock Genome sequence of the human malaria parasite plasmodium falciparum.
\newblock {\em Nature}, 419:498--511, 2002.

\bibitem[\protect\citeauthoryear{Hayakawa and
  Takemura}{2014}]{hayakawa2014estimation}
Jumpei Hayakawa and Akimichi Takemura.
\newblock Estimation of exponential-polynomial distribution by holonomic
  gradient descent, 2014.

\bibitem[\protect\citeauthoryear{Heiges \bgroup \em et al.\egroup
  }{2006}]{Heiges2006}
Mark Heiges, Haiming Wang, Edward Robinson, Cristina Aurrecoechea, Xin Gao,
  Nivedita Kaluskar, Philippa Rhodes, Sammy Wang, Cong-Zhou He, Yanqi Su, John
  Miller, Eileen Kraemer, and Jessica~C. Kissinger.
\newblock Cryptodb: a cryptosporidium bioinformatics resource update.
\newblock {\em Nucleic Acids Res.}, 34:D419--D422, 2006.

\bibitem[\protect\citeauthoryear{Hertz-Fowler \bgroup \em et al.\egroup
  }{2004}]{Hertz-Fowler2004}
Christiane Hertz-Fowler, Chris~S. Peacock, Valerie Wood, Martin Aslett, Arnaud
  Kerhornou, Paul Mooney, Adrian Tivey, Matthew Berriman, Neil Hall, Kim
  Rutherford, Julian Parkhill, Alasdair~C. Ivens, Marie-Adele Rajandream, and
  Bart Barrell.
\newblock Genedb: a resource for prokaryotic and eukaryotic organisms.
\newblock {\em Nucleic Acids Res.}, 32:D339--D343, 2004.

\bibitem[\protect\citeauthoryear{Hur \bgroup \em et al.\egroup
  }{2000}]{Hur2000}
A.~B. Hur, D.~Horn, H.~T. Siegelmann, and V.~Vapnik.
\newblock A support vector method for clustering.
\newblock {\em NIPS}, pages 367--373, 2000.

\bibitem[\protect\citeauthoryear{Hur \bgroup \em et al.\egroup
  }{2001}]{Hur2001}
A.~B. Hur, D.~Horn, H.~T. Siegelmann, and V.~Vapnik.
\newblock Support vector clustering.
\newblock {\em JMLR}, 2:125--137, 2001.

\bibitem[\protect\citeauthoryear{Kishore \bgroup \em et al.\egroup
  }{2013}]{Kishore}
S.~P. Kishore, J.~W. Stiller, and K.~W. Deitsch.
\newblock Horizontal gene transfer of epigenetic machinery and evolution of
  parasitism in the malaria parasite plasmodium falciparum and other
  apicomplexans.
\newblock {\em Evol Biol}, pages 13--37, 2013.

\bibitem[\protect\citeauthoryear{Koyama \bgroup \em et al.\egroup
  }{2014}]{koyama2014software}
Tamio Koyama, Hiromasa Nakayama, Katsuyoshi Ohara, Tomonari Sei, and Nobuki
  Takayama.
\newblock Software packages for holonomic gradient method.
\newblock In {\em Mathematical Software--ICMS 2014}, pages 706--712. Springer,
  2014.

\bibitem[\protect\citeauthoryear{Kuo \bgroup \em et al.\egroup
  }{2008}]{kissinger}
C.~Kuo, J.~P. Wares, and J.~C. Kissinger.
\newblock The apicomplexan whole-genome phylogeny: An analysis of incongruence
  among gene trees.
\newblock {\em Mol Biol Evol}, 25(12):2689--2698, 2008.

\bibitem[\protect\citeauthoryear{Marumo \bgroup \em et al.\egroup
  }{2014}]{marumo2014properties}
Naoki Marumo, Toshinori Oaku, and Akimichi Takemura.
\newblock Properties of powers of functions satisfying second-order linear
  differential equations with applications to statistics, 2014.

\bibitem[\protect\citeauthoryear{Nakayama \bgroup \em et al.\egroup
  }{2011}]{nakayama2011holonomic}
Hiromasa Nakayama, Kenta Nishiyama, Masayuki Noro, Katsuyoshi Ohara, Tomonari
  Sei, Nobuki Takayama, and Akimichi Takemura.
\newblock Holonomic gradient descent and its application to the fisher-bingham
  integral.
\newblock {\em Advances in Applied Mathematics}, 47(3):639 -- 658, 2011.

\bibitem[\protect\citeauthoryear{Owen and Provan}{2011}]{owen2011fast}
M.~Owen and J.~S. Provan.
\newblock A fast algorithm for computing geodesic distances in tree space.
\newblock {\em IEEE ACM T COMPUT BI}, 8(1):2--13, 2011.

\bibitem[\protect\citeauthoryear{Pain \bgroup \em et al.\egroup
  }{2005}]{Pain2005}
A.~Pain, H.~Renauld, and et~al.
\newblock Genome of the host-cell transforming parasite theileria annulata
  compared with t. parva.
\newblock {\em Science}, 309:131--133, 2005.

\bibitem[\protect\citeauthoryear{Weyenberg \bgroup \em et al.\egroup
  }{2014}]{weyenberg2014nonparametric}
G.~Weyenberg, P.~Huggins, C.~Schardl, D.K. Howe, and R.~Yoshida.
\newblock kdetrees: Nonparametric estimation of phylogenetic tree
  distributions.
\newblock {\em Bioinformatics}, 2014.

\bibitem[\protect\citeauthoryear{Zeilberger}{1990}]{zeilberger1990holonomic}
Doron Zeilberger.
\newblock A holonomic systems approach to special functions identities.
\newblock {\em Journal of Computational and Applied Mathematics}, 32(3):321 --
  368, 1990.

\end{thebibliography}
